\documentclass[10pt,journal,english,twocolumn]{IEEEtran} %
\usepackage{babel}
\setlocalecaption{english}{table}{\MakeUppercase{Table}}
\usepackage[babel]{microtype}
\usepackage[babel]{csquotes}

\usepackage[T1]{fontenc}
\usepackage[largesc]{newtxtext}
\usepackage[svgnames]{xcolor}
\usepackage{graphicx}
\usepackage{tikz}
\usepackage{pgfplots}
\usepackage{subfigure}
\usepackage{stfloats}

\usepackage{tabularx}
\usepackage{booktabs}
\usepackage{diagbox}

\usepackage{amsmath}
\usepackage{amsfonts}
\usepackage{amssymb}
\usepackage{amsthm}
\usepackage{bm}
\usepackage{siunitx}
\usepackage{mathtools}
\usepackage{dsfont}

\usepackage[math]{blindtext}
\usepackage{cancel}
\usepackage{titlesec}

\usepackage[backend=biber,url=false,style=ieee,bibstyle=ieee-proc,dashed=false,isbn=false,doi=true]{biblatex}
\renewcommand*{\bibfont}{\small} %

\usepackage{hyperref}
\hypersetup{
	pdfauthor={Karl-Ludwig Besser, Rafael F. Schaefer, H. Vincent Poor},
	pdftitle={Building Resilience in Wireless Communication Systems With a Secret-Key Budget},
	bookmarksnumbered=true,
	colorlinks=true,
	allcolors=.,
	urlcolor=MidnightBlue,
	pdflang={en-US},
}
\usepackage{bookmark}%
\bookmarksetup{%
	open,%
	openlevel=2,%
	addtohook={%
		\ifnum\bookmarkget{level}=1%
		\bookmarksetup{bold}%
		\fi%
	},%
}
\usepackage[all]{hypcap}

\usepackage[acronyms,nomain,xindy,nowarn]{glossaries}
\makeglossaries
\newacronym{5g}{5G}{Fifth-Generation Wireless Systems}
\newacronym{ai}{AI}{Artificial Intelligence}
\newacronym{ann}{ANN}{Artificial Neural Network}
\newacronym{ask}{ASK}{amplitude-shift keying}
\newacronym[plural={AWGN}, \glsshortpluralkey={AWGN}]{awgn}{AWGN}{additive white Gaussian noise}

\newacronym{bec}{BEC}{binary erasure channel}
\newacronym{ber}{BER}{bit error rate}
\newacronym{bler}{BLER}{block error rate}
\newacronym{bpsk}{BPSK}{binary phase-shift keying}

\newacronym{cdf}{CDF}{cumulative distribution function}
\newacronym{cm}{CM}{channel model}
\newacronym{cnn}{CNN}{convolutional neural network}
\newacronym{csi}{CSI}{channel state information}
\newacronym{csit}{CSI\babelhyphen{nobreak}T}{channel-state information at the transmitter}
\newacronym{csir}{CSI\babelhyphen{nobreak}R}{channel-state information at the receiver}

\newacronym{ecc}{ECC}{error-correcting code}
\newacronym{ecdf}{ECDF}{empirical distribution function}
\newacronym{ee}{EE}{energy efficiency}

\newacronym{ff}{FF}{feed-forward}
\newacronym{ffnn}{FF-NN}{feed-forward neural network}
\newacronym{fft}{FFT}{fast Fourier transform}
\newacronym{fsk}{FSK}{frequency-shift keying}

\newacronym{gm}{GM}{Gaussian mixture}

\newacronym{iid}{i.i.d.}{independent and identically distributed}
\newacronym{il}{IL}{information leakage}
\newacronym{iot}{IoT}{internet of things}

\newacronym{lb}{LB}{lower bound}
\newacronym{los}{LoS}{line-of-sight}

\newacronym{mac}{MAC}{multiple access channel}
\newacronym{map}{MAP}{maximum a posteriori}
\newacronym{mc}{MC}{Monte Carlo}
\newacronym{mimo}{MIMO}{multiple-input multiple-output}
\newacronym{miso}{MISO}{multiple-input single-output}
\newacronym{ml}{ML}{machine learning}
\newacronym{mmwave}{mmWave}{millimeter wave}
\newacronym{mop}{MOP}{multi-objective programming problem}
\newacronym{mos}{MOS}{mean opinion score}
\newacronym{mrc}{MRC}{maximum ratio combining}
\newacronym{msc}{MSC}{modulation and coding scheme}
\newacronym{mse}{MSE}{mean squared error}

\newacronym{nlos}{NLoS}{non-line-of-sight}
\newacronym{nn}{NN}{neural network}
\newacronym{nve}{NVE}{normalized validation error}

\newacronym{pdf}{PDF}{probability density function}
\newacronym{pls}{PLS}{physical layer security}
\newacronym{ppo}{PPO}{proximal policy optimization}
\newacronym{pwc}{PWC}{polar wiretap code}

\newacronym{qam}{QAM}{quadrature amplitude modulation}

\newacronym{ra}{RA}{rearrangement algorithm}
\newacronym{rbf}{RBF}{radial basis function}
\newacronym{relu}{ReLU}{rectified linear unit}
\newacronym{ris}{RIS}{reconfigurable intelligent surface}
\newacronym{rl}{RL}{reinforcement learning}
\newacronym{rnn}{RNN}{recurrent neural network}

\newacronym{sac}{SAC}{soft actor-critic}
\newacronym{sgd}{SGD}{stochastic gradient descent}
\newacronym{sgp}{SGP}{secure goodput}
\newacronym{simo}{SIMO}{single-input multiple-output}
\newacronym{sinr}{SINR}{signal-to-interference-plus-noise ratio}
\newacronym{siso}{SISO}{single-input single-output}
\newacronym{sk}{SK}{secret-key}
\newacronym{skg}{SKG}{secret-key generation}
\newacronym{sm}{SM}{source model}
\newacronym{snr}{SNR}{signal-to-noise ratio}
\newacronym{svm}{SVM}{support vector machine}

\newacronym{tvd}{TVD}{total variation distance}
\newacronym{tx}{TX}{transmission}

\newacronym{uav}{UAV}{unmanned aerial vehicle}
\newacronym{ub}{UB}{upper bound}
\newacronym{urllc}{URLLC}{ultra-reliable low latency communication}

\newacronym{v2v}{V2V}{vehicle-to-vehicle}
\newacronym{v2x}{V2X}{vehicle-to-everything}

\newacronym{zoc}{ZOC}{zero-outage capacity}
\setacronymstyle{long-short}
\glsdisablehyper

\definecolor{plot0}{HTML}{004488}
\definecolor{plot1}{HTML}{DDAA33}
\definecolor{plot2}{HTML}{BB5566}
\definecolor{plot3}{HTML}{000000}
\definecolor{plot4}{HTML}{AAAAAA}

\definecolor{plot5}{HTML}{00BB44}

\newcommand*{\inv}[1]{\ensuremath{#1^{-1}}}
\newcommand*{\positive}[1]{\ensuremath{\left[#1\right]^{+}}}
\newcommand*{\negative}[1]{\ensuremath{\left[#1\right]^{-}}}
\newcommand*{\avg}[1]{\ensuremath{\bar{#1}}}

\newcommand*{\diff}{\mathop{}\!\mathrm{d}}

\DeclareMathOperator*{\argmax}{arg\,max}

\DeclareMathOperator{\Ei}{Ei}

\newcommand*{\expect}[2][]{\ensuremath{\mathbb{E}_{#1}\left[#2\right]}}

\newcommand*{\berndist}{\ensuremath{\mathcal{B}}}

\newcommand*{\poisdist}{\ensuremath{\mathrm{Pois}}}

\let\rv\bm

\newcommand*{\cdf}{\ensuremath{F}}
\newcommand*{\X}{\ensuremath{\rv{X}}}

\newcommand*{\Y}{\ensuremath{\rv{Y}}}
\newcommand*{\ratebob}{\ensuremath{{R_{\textnormal{B}}}}}
\newcommand*{\rateeve}{\ensuremath{{R_{\textnormal{E}}}}}

\newcommand*{\rateskg}[1][]{\ignorespaces
	\if\relax\detokenize{#1}\relax
	{\ensuremath{\rv{R}_{\textnormal{SK}}}}%
	\else
	{\ensuremath{R_{\textnormal{SK}, #1}}}%
	\fi
}

\newcommand*{\ly}{\ensuremath{\lambda_{\textnormal{E}}}}

\newcommand*{\channelbob}{\ensuremath{\rv{H}_{\textnormal{B}}^2}}
\newcommand*{\channeleve}{\ensuremath{\rv{H}_{\textnormal{E}}^2}}

\newcommand*{\txpower}{\ensuremath{P_{\textnormal{T}}}}
\newcommand*{\txpowermax}{\ensuremath{P_{\textnormal{T,max}}}}

\newcommand*{\messagelength}{\ensuremath{L}}

\let\budget\surplus
\newcommand*{\netusage}{\ensuremath{\rv{Z}}}
\newcommand*{\netsum}{\ensuremath{\rv{S}}}
\newcommand*{\probruin}{\ensuremath{\psi}}
\newcommand*{\probsurv}{\ensuremath{\bar{\psi}}}
\newcommand*{\initbudget}{\ensuremath{b_0}}

\newcommand*{\probtx}{\ensuremath{p}}
\newcommand*{\probtxcrit}{\ensuremath{p_{\textnormal{crit}}}}
\newcommand*{\txindicator}{\ensuremath{\rv{M}}}

\newcommand*{\alertoutprob}{\ensuremath{\varepsilon}}
\newcommand*{\alertoutprobmax}{\ensuremath{\tilde{\varepsilon}}}
\newcommand*{\minbudget}{\ensuremath{b_{\alertoutprobmax}}}

\newcommand*{\alertduration}{\ensuremath{\rv{T}}}

\newcommand*{\resiloutageprob}{\ensuremath{\alpha}}
\newcommand*{\resiloutageprobmax}{\ensuremath{\tilde{\resiloutageprob}}}

\newcommand*{\weight}{\ensuremath{w}}
\newcommand*{\adaptee}{\ensuremath{g}}
\newcommand*{\reward}{\ensuremath{r}}
\pgfplotsset{compat=newest}
\pgfplotscreateplotcyclelist{lineplot cycle}{ %
	{plot0, mark=*, thick, mark options=solid},
	{plot1, mark=triangle*, thick, mark options=solid},
	{plot2, mark=square*, thick, mark options=solid},
	{plot3, mark=diamond*, thick, mark options=solid},
	{plot4, mark=pentagon*, thick, mark options=solid},
	{plot5, mark=triangle, thick, mark options=solid},
}

\pgfplotsset{%
	betterplot/.style={
		width=.93\linewidth,
		height=.27\textheight,
		xlabel near ticks,
		ylabel near ticks,
		cycle list name=lineplot cycle,
		mark options=solid,
		xmajorgrids=true,
		xminorgrids=true,
		ymajorgrids=true,
		grid style={line width=.1pt, draw=gray!20},
		major grid style={line width=.25pt,draw=gray!30},
		legend cell align=left,
		legend style = {
			/tikz/every even column/.append style={column sep=0.33cm}
		},
	},
}
\usetikzlibrary{positioning,calc,external,fit}
\usepgfplotslibrary{fillbetween}
\tikzset{
block/.style={draw},
}

\newcommand{\todo}[2][]{\ignorespaces\leavevmode
	\if\relax\detokenize{#1}\relax
	{\color{red}[TODO: #2]}%
	\else
	{\color{red}[TODO (#1): #2]}%
	\fi
}

\definecolor{change}{HTML}{0096b8}

\theoremstyle{plain}%
\newtheorem{thm}{Theorem}

\newtheorem{lem}{Lemma}
\newtheorem{prop}{Proposition}

\theoremstyle{definition}

\newtheorem*{prob*}{Problem Formulation}

\theoremstyle{remark}
\newtheorem{rem}{Remark}
\newtheorem*{rem*}{Remark}

\newtheoremstyle{example}{\topsep}{\topsep}{}{}{\itshape}{.}{ }{}
\theoremstyle{example}
\newtheorem{example}{Example}
\newtheorem*{example*}{Example}

\addto\extrasenglish{%

}

\IfPackageLoadedTF{titlesec}{%
	\ifCLASSOPTIONconference%
	\titlespacing{\section}{0pt}{1.5ex plus 1.5ex minus 0.5ex}{0.7ex plus 1ex minus 0ex} %
	\titlespacing{\subsection}{0pt}{1.5ex plus 1.5ex minus 0.5ex}{0.7ex plus .5ex minus 0ex} %
	\else
	\titlespacing{\section}{0pt}{3.0ex plus 1.5ex minus 1.5ex}{0.7ex plus 1ex minus 0ex} %
	\titlespacing{\subsection}{0pt}{3.5ex plus 1.5ex minus 1.5ex}{0.7ex plus .5ex minus 0ex} %
	\fi
	
	\def\thesubsubsectiondis{\arabic{subsubsection})}
	\titleformat{\subsubsection}[runin]{\itshape}{\thesubsubsectiondis}{.5em}{}[:]
	\titlespacing*{\subsubsection}{\parindent}{0ex plus 0.1ex minus 0.1ex}{1ex}
}{}

\title{Building Resilience in Wireless Communication Systems With a Secret-Key Budget}
\author{%
Karl-Ludwig Besser, \IEEEmembership{Member, IEEE}, Rafael F. Schaefer, \IEEEmembership{Senior Member, IEEE}, and\\H. Vincent Poor, \IEEEmembership{Life Fellow, IEEE}%
\thanks{Parts of this work were presented at the 2024 IEEE International Symposium on Personal, Indoor and Mobile Radio Communications (PIMRC)~\cite{Besser2024pimrc}.}
\thanks{Karl-Ludwig Besser was with the Department of Electrical and Computer Engineering, Princeton University, Princeton, NJ 08544, USA, and is now with the Department of Electrical Engineering, Linköping University, Sweden (email: karl-ludwig.besser@liu.se).
Rafael F. Schaefer is with the Chair of Information Theory and Machine Learning, the BMBF~Research Hub 6G-life, the Cluster of Excellence \enquote{Centre for Tactile Internet with Human-in-the-Loop} (CeTI), and the 5G Lab Germany, Technische Universität Dresden, 01062 Dresden, Germany (e-mail: rafael.schaefer@tu-dresden.de).
H. Vincent Poor is with the Department of Electrical and Computer Engineering, Princeton University, Princeton, NJ 08544, USA (email: poor@princeton.edu).%
}%
\thanks{The work of K.-L.~Besser is supported by Security~Link and was supported by the German Research Foundation (DFG) under grant~BE\,8098/1-1.
The work of R. F. Schaefer is supported in part by the German Federal Ministry of Education and Research (BMBF) within the National Initiative on 6G~Communication Systems through the Research Hub \emph{6G-life} under grant~16KISK001K, in part by the BMBF within the 6G-ANNA project under grant~16KISK103, in part by the DFG as part of Germany's Excellence Strategy~EXC\,2050/1 (Project ID~390696704), Cluster of Excellence \emph{\enquote{Centre for Tactile Internet with Human-in-the-Loop} (CeTI)} at Technische Universität Dresden, and in part by the DFG under Grant~SCHA1944/11-1.
The work of H.~V. Poor is supported by the U.S.\ National Science Foundation under grant~ECCS-2335876.}%
}

\begin{document}
\maketitle

\begin{abstract}\noindent\boldmath
Resilience and power consumption are two important performance metrics for many modern communication systems, and it is therefore important to define, analyze, and optimize them.
In this work, we consider a wireless communication system with secret-key generation, in which the secret-key bits are added to and used from a pool of available key bits.
We propose novel physical layer resilience metrics for the survivability of such systems.
In addition, we propose multiple power allocation schemes and analyze their trade-off between resilience and power consumption.
In particular, we investigate and compare constant power allocation, an adaptive analytical algorithm, and a reinforcement learning-based solution.
It is shown how the transmit power can be minimized such that a specified resilience is guaranteed.
These results can be used directly by designers of such systems to optimize the system parameters for the desired performance in terms of reliability, security, and resilience.
\end{abstract}
\begin{IEEEkeywords}%
	Resilience,
	power control,
	secret-key budget,
	physical layer security,
	ruin theory.
\end{IEEEkeywords}
\glsresetall

\section{Introduction}\label{sec:introduction}
Resilience of a (communication) system refers to its capability to function acceptably in case of failures and malfunctions.
It is an important performance metric for many modern communication networks~\cite{Sterbenz2010}.
In particular, mission-critical systems~\cite{EricsonMissionCritical}, such as smart railways, public safety, and smart power grid systems~\cite{Panteli2015,Moreno2020,Stankovi2023,Ji2017,Xu2024}, require resilience as failures could result in severe damages or even risk people's safety.

In addition to resilience, modern wireless communication systems need to fulfill other, often conflicting, objectives at the same time.
Reliability, resilience, privacy, and security should all be maximized while simultaneously conserving energy and minimizing transmit power.
Especially the security of data transmissions is crucial as more sensitive data is transmitted in modern applications~\cite{Nguyen2021}.

Traditionally, cryptography is used to ensure data confidentiality and integrity.
However, as wireless networks become more prevalent and complex, classical cryptographic approaches face many challenges.
This includes quantum attacks~\cite{Gidney2021howtofactorbit}, key management issues, and computational costs for embedded and small devices, e.g., in the context of the \gls{iot}.
An alternative approach to combat these challenges can be \gls{pls}, which leverages the unique physical properties of the wireless communication channel to establish security~\cite{Bloch2011,Poor2017}.
One aspect of \gls{pls} is \gls{skg} in which the channel can be utilized to securely generate key bits, which then act as a one-time pad to encrypt a transmitted message.
The main idea behind \gls{skg} is to use channel reciprocity to establish a source of common randomness between the transmitter and the legitimate receiver, which can be used to distill key bits~\cite{Li2019skg}.
Different schemes for an implementation of \gls{skg} have been proposed in the literature~\cite{KameniNgassa2017,Li2018skg,Aono2005}.
Besides achieving perfect information-theoretic security, using \gls{skg} has the advantage of being less dependent on complicated key management issues such as key distribution and revocation.
Additionally, using generated key bits as one-time pads provides a built-in resistance against key compromise as each key bit can only be used once.
However, because of this, it is also essential to ensure that the (legitimate) communication parties always have a sufficient amount of key bits available whenever a message is to be transmitted.

This reliability in terms of the probability of running out of \gls{sk} bits, as well as the latency to restore a certain amount of bits, is investigated in~\cite{Besser2024SKGbudget}.
In particular,~\cite{Besser2024SKGbudget}~introduces the perspective of modeling the problem as treating the amount of available key bits as a budget with opposing income and spending processes.
New bits are generated using \gls{skg} techniques and added to the budget, whereas using them to encrypt and transmit messages removes them from the pool of available bits.

In this work, we adapt the same notion of a \gls{sk} budget to analyze the aforementioned communication systems using \gls{skg}.
However, in contrast to~\cite{Besser2024SKGbudget}, which only focuses on reliability, we focus on the \emph{resilience} of such systems in this work.
While resilience has many aspects, it has been mostly considered from a network perspective in communications~\cite{Rak2015,Mauthe2016,Menth2009,AlAqqad2022,Esmat2023}, e.g., in terms of failing links in routing packets through a network.
In contrast, we focus on resilience from a physical layer perspective in this work. %
In particular, we investigate the ability to withstand emergency events and derive necessary preparation steps.
For this, we introduce novel resilience metrics and analyze the influence of the transmit power and highlight the trade-off between energy conservation and resilience.
These metrics are an important step for quantifying the physical layer resilience of modern communication systems.
Existing metrics~\cite{Cholda2009,Hosseini2016} are either focused on a network layer perspective, e.g., packet loss ratio and \gls{mos}, or very general with only broad system level definitions.
A metric for jointly assessing mixed criticality and resilience for physical layer resource management is proposed in \cite{Reifert2023}.
Other previous works that consider resilience on the physical layer often treat resilience as the ability to mitigate possible jamming attacks~\cite{Okyere2020,Restuccia2022,Letafati2020}.

The main contributions of this work are summarized as follows.
\begin{itemize}
	\item To the best of the authors' knowledge, this work is the first comprehensive study of the physical layer resilience of wireless communication systems with secret-key generation in terms of their survivability. %
	\item We review the classical resilience model and adapt it to wireless communication systems with a secret-key budget.
	In particular, we introduce a notion of resilience for such systems and present specific resilience metrics (\autoref{sub:resilience-model} and \autoref{sub:resilience-metric}).
	\item We analyze the resilience and the trade-off between transmit power and resilience of a communication system with a secret-key budget for different power allocation schemes (\autoref{sec:resilience-analysis}).
	\item We illustrate the general results with numerical examples, highlighting the practical use and insights for a system designer (\autoref{sec:numerical-example}).
\end{itemize}
Additionally, the source code to reproduce all presented results and simulations is publicly available at~\cite{BesserGithub}.

\subsubsection*{Notation}
Random variables are denoted in capital boldface letters, e.g.,~$\X$, and their realizations in small letters, e.g.,~$x$.
We use~$F_{\X}$ and~$f_{\X}$ for the probability distribution and its density, respectively.
The expectation is denoted by~$\mathbb{E}$ and the probability of an event by~$\Pr$.
The Bernoulli distribution with mean~$p$ is denoted as~$\berndist(p)$.
As a shorthand, we use ${\positive{x}=\max\left\{x, 0\right\}}$ and ${\negative{x} = \min\left\{x, 0\right\}}$.
An overview of the most commonly used variable notation can be found in \autoref{tab:notation-variables}.
\begin{table}%
	\renewcommand*{\arraystretch}{1.25}
	\centering
	\caption{Definitions of the Most Commonly Used Variables}\label{tab:notation-variables}
	\begin{tabularx}{.9\linewidth}{lX}
		\toprule
		$\channelbob$ & Channel gain to Bob (includes beamforming)\\
		$\channeleve$ & Channel gain to Eve (includes beamforming)\\
		$\txpower$ & Transmit power for \gls{skg}\\
		$\budget$ & Number of available \gls{sk} bits\\
		$\initbudget$ & Number of initially available \gls{sk} bits\\
		$\messagelength$ & Message length\\
		$\rateskg$ & \Gls{skg} rate\\
		$\netusage$ & \Gls{sk} bit usage\\
		$\txindicator$ & Binary random variable that indicates whether a message needs to be transmitted\\
		$\probtx$ & Probability that a message is transmitted\\
		$\alertduration$ & Duration of the alert state\\
		$\probruin_{\initbudget}$ & Ruin probability with initial budget~$\initbudget$\\
		$\probsurv_{\initbudget}$ & Survival probability with initial budget~$\initbudget$\\
		$\alertoutprob$ & Outage probability when entering the alert state\\
		$\alertoutprobmax$ & Target maximum outage probability for the alert state\\
		$\resiloutageprob$ & Resilience outage probability\\
		\bottomrule
	\end{tabularx}
\end{table}
\section{Preliminaries and Background}\label{sec:preliminaries}
In this section, we will review important existing ideas and results from secret-key generation and ruin theory, which will be needed as background information for the remainder of this work.

\subsection{Secret-Key Generation}\label{sub:intro-skg}
One way to achieve perfect information-theoretic secrecy is the use of a one-time pad~\cite{Shannon1949}.
The concept of a one-time pad involves using key bits that are known exclusively to the legitimate communication parties.
Secret-key generation allows these parties to agree on such secret key bits by leveraging the physical properties of their communication channel~\cite{Bloch2021}.
Various models and algorithms have been proposed and analyzed in the literature for different communication scenarios~\cite{KameniNgassa2017,Li2018skg,Aono2005,Furqan2016,Gao2021,Linh2023}, including static environments~\cite{Aldaghri2020}, quasi-static fading channels~\cite{Renna2013}, and fast-fading channels with correlated channels~\cite{Zorgui2016,Besser2020wsa}.

The fundamental idea behind secret-key agreement is that Alice and Bob share access to a common randomness, which they can use to extract identical key bits~\cite{Maurer1993}.
To correct errors and protect against eavesdroppers, they exchange messages over a public channel.
Prominent models for secret-key agreement include the \gls{cm} and the \gls{sm} schemes~\cites[Chap.~4]{Bloch2011}{Lai2014}, each using distinct approaches to generating random observations.
In the \gls{sm} model~\cite{Wallace2010,Maurer1993}, both legitimate nodes and potential eavesdroppers access a common source of randomness described by a joint distribution.
Conversely, in the \gls{cm} scheme, randomness stems from transmissions over a noisy wiretap channel~\cite{Ahlswede1993}.
Utilizing the public channel, secret-key bits are derived from observed randomness at the legitimate nodes in such a way that the eavesdropper remains uninformed.

In this work, we do not focus on any particular \gls{skg} scheme or the generation of specific key bits.
Instead, we investigate the number of available key bits over time and associated resilience notions.
In particular, we use the concept of a secret-key budget~\cite{Besser2024SKGbudget}.
As described above, the legitimate nodes perform \gls{skg} and agree on secret-key bits using a \gls{skg} scheme.
These newly generated bits are appended to the pool of existing key bits, establishing a reservoir of available \gls{sk} bits at both legitimate communicating parties.
During secure transmission of a message with length~$\messagelength$, the oldest $\messagelength$~key bits are utilized as a one-time pad for encryption.
Given that solely the legitimate transmitter and receiver possess knowledge of the key bits, the transmission achieves information-theoretic security.
However, due to the nature of one-time pads, the \gls{sk} bits employed for transmission can only be used once and are therefore removed from the list of available key bits.

\subsection{Ruin Theory}\label{sub:intro-ruin-theory}
Ruin theory, originally rooted in economics and actuarial science, serves as an analytical tool for assessing the solvency of insurance companies~\cite{Asmussen2020}.
At its core, this theory deals with the two conflicting cash flows within an insurance company.
On the one hand, there is an income stream from the premiums paid by the customers.
Conversely, claims arise sporadically, which decrease the company's budget.

In the classical model, premiums arrive at a constant positive rate, while claims occur randomly, often modeled by a Poisson process.
The main quantity of interest is the probability that the insurance company will go bankrupt.

For this, we define the time of ruin~$\tau$, denoting the first time~$t$ at which the budget~$\budget$ falls to zero~\cite{Dickson2016}, i.e., ${\tau = \inf \{t \geq 0 \;|\; \budget(t) \leq 0 \}}$.
The probability that ruin occurs before a given time~$t$ is referred to as the \emph{ruin probability}~$\probruin_{\initbudget}$,%
\begin{equation}\label{eq:def-prob-ruin}
	\probruin_{\initbudget}(t)=\Pr(\tau \leq t)\,,
\end{equation}
where the insurance company starts with an initial budget~$\initbudget$.
The \emph{survival probability}~$\probsurv_{\initbudget}(t)=1-\probruin_{\initbudget}(t)$ is defined as the complementary event, i.e., the probability that the system's budget does not go to zero within $t$~time slots.
Similarly the \emph{probability of ultimate ruin}~$\probruin_{\initbudget}(\infty)$ is given as the limit that ruin will eventually occur, i.e., $\probruin_{\initbudget}(\infty) = \Pr(\tau < \infty)$.

Ruin theory has been extensively studied in the literature of mathematical finance and actuarial science, where many different problems have been discussed and explicitly solved, including expressions for finite-time ruin probability when considering specific claim distributions~\cite{Chan2006,Picard1997} or when considering interest rates~\cite{Cai2002,Wang2006interest}.
For a comprehensive exposition on the subject, we refer the reader to~\cite{Asmussen2020,Dickson2016}.

The system model considered in this work shows parallels to the insurance company model in classical ruin theory.
However, due to inherent differences between financial and communication systems, we cannot directly apply existing results from ruin theory.
Instead, we will leverage concepts and definitions from ruin theory and adapt them to the specific problem considered in this work.
The system model and the exact problem formulation are discussed in the following section.
\section{System Model and Problem Formulation}\label{sec:system-model}

Throughout this work, we consider a communication system, in which a transmitter~(Alice) wants to securely transmit data to a legitimate receiver~(Bob).
The transmission is overheard by a passive eavesdropper~(Eve).
We assume that all channels are quasi-static fading channels with additional \gls{awgn} at the receivers~\cite[Chap.~5.2]{Bloch2011}, i.e., we assume that the channels remain constant for the transmission of one codeword.
Throughout the following, we normalize the noise power to one to simplify the notation.
The channel gains of the channel between Alice and Bob, and the channel between Alice and Eve in time slot~$t$ are denoted as $\channelbob(t)$ and $\channeleve(t)$, respectively.
They are assumed to be independent random variables, which also change independently over time.
Since we consider a stationary scenario, we assume that beamforming and antenna gains are constant over time.
As they simply scale the channel gains~$\channelbob$ and~$\channeleve$, we assume the distributions of these channel gains to include all additional effects like beamforming and antenna gains.
The transmitter transmits at power level~$\txpower(t)$ in time slot~$t$.
Thus, the \gls{snr} values of the main channel and the eavesdropper's channel at time~$t$ are given as $\X(t)=\txpower(t)\channelbob(t)$ and $\Y(t)=\txpower(t)\channeleve(t)$, respectively.

Based on this model, we can achieve the following rates to Bob and Eve at time~$t$
\begin{align}
	\rv{\ratebob}(t) &= \log_2\big(1+\X(t)\big)\\
	\rv{\rateeve}(t) &= \log_2\big(1+\Y(t)\big)\,,
\end{align}
respectively.
Furthermore, we assume that the system has a maximum power constraint~$\txpowermax$, i.e., ${\txpower(t)\leq\txpowermax}$. %

An overview of the system model can be found in \autoref{fig:system-model}.
An application example is later discussed in detail in \autoref{sub:application-examples}.

\begin{figure}
	\centering
	\begin{tikzpicture}[font=\small]

\node[coordinate] (message) {};
\node[block,right=of message] (xorAlice) {XOR};
\draw[->] (message) -- node[above=.1] {Message} (xorAlice);

\node[block,above=.75 of xorAlice] (budgetAlice) {SK Budget $\budget$};
\draw[->] (budgetAlice) -- node[right] {Key} (xorAlice);

\node[block,right=.5 of xorAlice] (tx) {TX};

\node[above=.3 of tx,inner sep=0pt] (pt) {$\txpower$};
\draw[->] (pt) -- (tx);
\draw[->] (xorAlice) -- (tx);

\node[block,right=.7 of tx] (hb) {$\channelbob$};
\draw[->] (tx) -- (hb);
\node[block,below=.5 of hb] (he) {$\channeleve$};
\draw[->] ($(tx)!.5!(hb)$) |- (he);

\node[block,right=.7 of hb] (rx) {RX};
\draw[->] (hb) -- (rx);

\node[block,right=.5 of rx] (xorBob) {XOR};
\draw[->] (rx) -- (xorBob);

\node[block,above=.75 of xorBob] (budgetBob) {SK Budget $\budget$};
\draw[->] (budgetBob) -- node[left] {Key} (xorBob);
\draw[<->] (budgetAlice) -- node[above] {\Gls{skg}} (budgetBob);
\draw[->] (pt) -| ($(budgetAlice)!.5!(budgetBob)$);

\node[coordinate,right=of xorBob] (messageRec) {};
\draw[->] (xorBob) -- node[above=.1] {Message} (messageRec);

\node[draw,dashed,rounded corners,fit=(message)(tx)(budgetAlice),label={above:Alice}] (alice) {};
\node[draw,dashed,rounded corners,fit=(messageRec)(rx)(budgetBob),label={above:Bob}] (bob) {};

\node[right=of he] (messageEve) {\xcancel{Message}};
\draw[->] (he) -- (messageEve);

\node[draw,dashed,rounded corners,fit=(messageEve),label={below:Eve}] (eve) {};
\end{tikzpicture}
	\vspace*{-.5em}
	\caption{%
	Overview of the system model.
	Alice and Bob generate key bits and add them to the \gls{sk} budget.
	Before each transmission, key bits from the budget are used as one-time pad to encrypt the message, such that Eve is not able to decrypt it.
	For the sake of clarity, we omit the fact that the communication channels are also used for secret-key generation (SKG).
	Therefore, it is also affected by $\channelbob$, $\channeleve$, and $\txpower$, cf.~\eqref{eq:skg-rate}.
	}%
	\label{fig:system-model}
\end{figure}

\subsection{Secret-Key Generation and Secret-Key Budget}
For a secure data transmission, Alice and Bob use key bits as an one-time pad for encryption.
In order to generate these key bits confidentially, \gls{skg} techniques can be used.
By leveraging physical properties of the propagation channels, Alice and Bob can confidentially generate bits, which are only known to them and not to Eve.
For this work, we assume that the standard source model method for \gls{awgn} wiretap channels is used.
The \gls{skg} rate in time slot~$t$ for the considered model is given by~\cite[Chap.~5.1]{Bloch2011} %
\begin{equation}\label{eq:skg-rate}
	\rateskg(t) = \log_2\left(\frac{1+\X(t)+\Y(t)}{1+\Y(t)}\right)\,.
\end{equation}

Newly generated bits are added to a pool of available key bits.
Whenever a message of length~$\messagelength$ is transmitted, the oldest $\messagelength$ bits from this pool are used as a one-time pad~\cite{Lugrin2023}.
Because of this, each \gls{sk} bit can only be used exactly once and they are then removed from the set of available key bits.
This idea of adding and removing key bits implies the notion of a \emph{secret-key budget}, which keeps track of the number of available \gls{sk} bits at each time.
The development of this budget~$\budget$ over time can be mathematically described as
\begin{equation}\label{eq:def-budget}
	\budget(t) = \initbudget - \sum_{i=1}^{t} \netusage(i) = \initbudget - \netsum(t)\,,
\end{equation}
where $\initbudget=\budget(0)$ is the initial budget, $\netusage(i)$ is the usage of \gls{sk} bits in time slot~$i$, and $\netsum(t)=\sum_{i=1}^{t}\netusage(i)$ is the accumulated usage until time slot~$t$.

Throughout this work, we assume that messages with a constant length~$\messagelength$ arrive randomly with probability~$\probtx$ in each time slot.
If a message arrives in time slot~$t$, it is encrypted and immediately transmitted, i.e., removing $\messagelength$ bits from the budget.
If no message arrives, \gls{skg} is performed and new key bits are added to the budget.
This yields the following expression for the usage~$\netusage$
\begin{equation}\label{eq:def-net-usage}
	\netusage(t) =
	\begin{cases}
		-\rateskg(t) & \text{if } \txindicator(t) = 0\\
		\messagelength & \text{if } \txindicator(t) = 1
	\end{cases}
\end{equation}
where $\txindicator(t)\sim\berndist(\probtx)$ is the binary random variable that indicates whether a message arrived in time slot~$t$, i.e., $\probtx = \Pr\big(\txindicator(t) = 1\big)$.
Note that the \emph{usage}~$\netusage$ is negative whenever bits are added to the budget.
An illustration of the scheduling model can be found in \autoref{fig:model-time-scheme}.

\begin{figure}
	\centering
	\begin{tikzpicture}%
	\begin{axis}[
		width=.96\linewidth,
		height=.2\textheight,
		axis line style={draw=none},
		xlabel={Time Slot~$t$},
		xlabel near ticks,
		x tick label as interval,
		ylabel near ticks,
		ylabel={{\Gls{sk} Bit Usage~$\netusage(t)$}},
		ytick={0,.75},
		yticklabels={$0$, $\messagelength$},
		xmin=.99,
		xmax=10.2,
		xtick distance=1,
		ymin=-1,
		ymax=1.3,
		axis on top,
		]
		\plot[black!50,samples=2,domain=0:12] {0};
		\draw[plot2,thick,fill=plot2,fill opacity=.3] foreach \n in {1, ..., 4, 6, 7, 9}{(axis cs:\n,-1) rectangle (axis cs:\n+1,1)};
		\draw foreach \n in {1, ..., 4, 6, 7, 9}{(axis cs:\n+0.5, 1) node[anchor=south,font=\small] {SKG}};
		
		\draw[black,thick,dashed] foreach \n in {1, ..., 15}{(axis cs:\n, -2) -- (axis cs:\n, 2)};
		
		\draw[plot3,thick,fill=plot3,fill opacity=.3] foreach \n in {5, 8, 10}{(axis cs:\n,-1) rectangle (axis cs:\n+1,1)};
		\draw foreach \n in {5,8,10}{(axis cs:\n+0.5, 1) node[anchor=south,font=\small] {TX}};
		
		\addplot[black, ultra thick, const plot] table[x=x, y=y] {
			x	y
			1	-0.25
			2	-0.7
			3	-.1
			4	-.5
			5	.75
			6	-.4
			7	-.75
			8	.75
			9	-.33
			10	.75
			11	.75
		};
	\end{axis}
\end{tikzpicture}
	\caption{%
	Illustration of the scheduling model in the normal state.
	In each time slot~$t$, there is a probability~$\probtx$ that a message of length~$\messagelength$ is transmitted.
	If no message is transmitted, \gls{skg} is performed instead, which adds key bits to the \gls{sk} budget.
	}%
	\label{fig:model-time-scheme}
\end{figure}

\begin{example}[{\Gls{sk} Budget}]\label{ex:sk-budget}
In order to illustrate the idea of the \gls{sk} budget, we use the following example, which is also depicted in \autoref{fig:illustration-sk-budget}.
Up until time slot~$t$, Alice and Bob have already generated six key bits~$(0, 1, 1, 0, 0, 1)$, i.e., $\budget(t-1)=6$.
We now assume that \gls{skg} is performed in time slot~$t$ as illustrated in \autoref{fig:illustration-sk-budget-skg}.
This generates four new key bits~$(1, 0, 1, 1)$, i.e., $\netusage(t)=-4$.
It should be recalled that $\netusage$ describes the \emph{usage} of key bits, i.e., it is negative whenever new bits are generated.
The four new key bits are then appended to the existing key bits, increasing the total number of available key bits at the end of time slot~$t$ to $\budget(t)=10$.

In the next time slot~$t+1$, a message of length~$L=5$ arrives, which needs to be transmitted securely.
Therefore, Alice uses the oldest $L$~\gls{sk} bits as a one-time pad to encrypt the message, i.e., performing XOR of the message and key bits.
In the shown example in \autoref{fig:illustration-sk-budget-tx}, the resulting encrypted message~$(1, 0, 1, 1, 0)$ is then transmitted to Bob, who uses the same oldest $L$~\gls{sk} bits from the budget to decode the message.
In order to guarantee perfect secrecy, each one-time pad can only be used once.
Therefore, the used $L$~key bits are removed from the \gls{sk} budget at the end of time slot~$t+1$.
This reduces the number of available key bits to $\budget(t+1)=5$.

\begin{figure}
	\centering
	\subfigure[{Illustration of the \gls{skg} phase. In this example, four new key bits ($1$, $0$, $1$, $1$) are generated and added to the \gls{sk} budget.\label{fig:illustration-sk-budget-skg}}]{\centering\begin{tikzpicture}
	\foreach \i [count=\xi] in {0,1,1,0,0,1}{%
		\node[draw] (skAlice\xi) at (.5*\xi, 0) {\i};
		\node[draw] (skBob\xi) at (.5*\xi, -2) {\i};
	}
	\foreach \i [count=\xi] in {1, 0, 1, 1}{%
		\node[draw,plot2,thick] (skAliceNew\xi) at (.5*\xi+3, 0) {\i};
		\node[draw,plot2,thick] (skBobNew\xi) at (.5*\xi+3, -2) {\i};
		\node[draw,plot2] (skCenterNew\xi) at (.5*\xi+1.5, -1) {\i};
	}
	\node[left=.1 of skCenterNew1,outer sep=.2] (skg) {\Gls{skg}};
	
	\node[left=.1 of skAlice1] (skAlice) {\Gls{sk} Bits};
	\node[left=.1 of skBob1] (skBob) {\Gls{sk} Bits};
	
	\node[fit=(skAlice)(skAliceNew4),draw,thick,inner sep=.5em,label={Alice}] (alice) {};
	\node[fit=(skBob)(skBobNew4),draw,thick,inner sep=.5em,label={below:Bob}] (bob) {};

	\draw[<->,thick] (alice.south -| skg.west) -- (bob.north -| skg.west);
	\draw[<->,thick] ($(alice.south -| skCenterNew4.east) + (.2, 0)$) -- ($(bob.north -| skCenterNew4.east) + (.2, 0)$);
	
\end{tikzpicture}}
	\subfigure[{Illustration of the \gls{tx} phase. The message of length~$\messagelength=5$ is encrypted using the $\messagelength$ oldest key bits as a one-time pad.\label{fig:illustration-sk-budget-tx}}]{\centering\begin{tikzpicture}
	\foreach \i [count=\xi] in {0,1,1,0,0}{%
		\node[draw,plot2,densely dashed,thick] (skAlice\xi) at (.5*\xi, 0) {\i};
		\node[draw,plot2,densely dashed,thick] (skBob\xi) at (.5*\xi, -3.75) {\i};
		\node[draw,plot2] (skCenter\xi) at (.5*\xi+1.25, -1.75) {\i};
	}
	\node[left=.1 of skCenter1] {Key};
	
	\foreach \i [count=\xi] in {1,1,0,1,0}{%
		\node[draw] (message\xi) at (.5*\xi+1.25,-1) {\i};
	}
	\node[left=.1 of message1] {Message};

	\foreach \i [count=\xi] in {1,0,1,1,0}{%
		\node[draw] (tx\xi) at (.5*\xi+1.25,-2.75) {\i};
		\draw[->] (skCenter\xi) -- (tx\xi);
	}
	\node[left=.1 of tx1] {Transmission};

	\foreach \i [count=\xi] in {1, 1, 0, 1, 1}{%
		\node[draw] (skAliceNew\xi) at (.5*\xi+2.5, 0) {\i};
		\node[draw] (skBobNew\xi) at (.5*\xi+2.5, -3.75) {\i};
	}
	
	\node[left=.1 of skAlice1] (skAlice) {\Gls{sk} Bits};
	\node[left=.1 of skBob1] (skBob) {\Gls{sk} Bits};
	
	\node[fit=(skAlice)(skAliceNew5),draw,thick,inner sep=.5em,label={Alice}] (alice) {};
	\node[fit=(skBob)(skBobNew5),draw,thick,inner sep=.5em,label={below:Bob}] (bob) {};

\end{tikzpicture}}
	\caption{Illustration of the \gls{sk} budget idea. If \gls{skg} is performed, new key bits are generated and added to the budget at Alice and Bob. If a message of length~$\messagelength$ arrives, the oldest $\messagelength$~key bits are used as a one-time pad to encrypt it and removed from the budget. (\autoref{ex:sk-budget})}%
	\label{fig:illustration-sk-budget}
\end{figure}
\end{example}

\begin{rem}[{Connection Between Ruin Theory and Our Model}]
	In \autoref{sub:intro-ruin-theory}, we introduced some of the basics of ruin theory, an area traditionally used to analyze the solvency of insurance companies.
	In particular, the basic model in ruin theory considers two opposing cash flows, namely the income of insurance companies from premiums paid by customers on the one hand, and insurance claims paid back to customers on the other.
	While there are some differences in the details, we can still draw parallels to the system model considered in this work.
	In particular, the solvency of our system is given by the \gls{sk} budget~$\budget$, which also observes two opposing flows, namely the newly generated key bits using \gls{skg}, which correspond to the income of our system, and the removed bits due to their use as one-time pads, which corresponds to the payments reducing our budget.
	Therefore, we will leverage concepts and definitions from ruin theory and adapt them to the specific problem considered in this work.
\end{rem}

\subsection{Resilience Model}\label{sub:resilience-model}
\begin{figure*}
	\centering
	\begin{tikzpicture}%
	\begin{axis}[
		betterplot,
		width=.97\linewidth,
		height=.23\textheight,
		xlabel={Time~$t$},
		ylabel={Performance (\Gls{sk} Budget)},
		xmin=0,
		xmax=10,
		ymin=0,
		ymax=7,
		xtick={0, 2, 3.5, 5.5, 8, 10},
		xticklabels={$0$, $t_1$, $t_2$, $t_3$, $t_4$, {}},%
		ytick={0, 1, ..., 8},
		yticklabels={},
		]
		\addplot+[restrict x to domain=0:2] table [x=t,y=x] {data/samples-resilience-model.dat};
		\addplot+[restrict x to domain=2:3.5] table [x=t,y=x] {data/samples-resilience-model.dat};
		\addplot+[restrict x to domain=3.5:5.5] table [x=t,y=x] {data/samples-resilience-model.dat};
		\addplot+[restrict x to domain=5.5:8] table [x=t,y=x] {data/samples-resilience-model.dat};
		\pgfplotsset{cycle list shift=-4};
		\addplot+[restrict x to domain=8:10] table [x=t,y=x] {data/samples-resilience-model.dat};

		\draw[plot0,fill,fill opacity=.15] (axis cs:0, 0) rectangle (axis cs:2, 7);
		\draw[plot1,fill,fill opacity=.15] (axis cs:2, 0) rectangle (axis cs:3.5, 7);
		\draw[plot2,fill,fill opacity=.15] (axis cs:3.5, 0) rectangle (axis cs:5.5, 7);
		\draw[plot3,fill,fill opacity=.15] (axis cs:5.5, 0) rectangle (axis cs:8, 7);
		\draw[plot0,fill,fill opacity=.15] (axis cs:8, 0) rectangle (axis cs:10, 7);

		\draw[|<->|, thick] (axis cs:0,6) -- (axis cs:2,6);
		\node[fill=white, font=\normalsize] at (axis cs:1,6) {Normal State};
		
		\draw[|<->|, thick] (axis cs:2,6) -- (axis cs:3.5,6);
		\node[fill=white, font=\normalsize] at (axis cs:2.75,6) {Alert State};
		\draw[|<->|, thick] (axis cs:2,.5) -- (axis cs:3.5,.5);
		\node[fill=white, font=\normalsize] at (axis cs:2.75,.5) {$\alertduration$};
		
		\draw[|<->|, thick] (axis cs:3.5,6) -- (axis cs:5.5,6);
		\node[fill=white, font=\normalsize,align=center] at (axis cs:4.5,6) {Degraded\\Performance};
		
		\draw[|<->|, thick] (axis cs:5.5,6) -- (axis cs:8,6);
		\node[fill=white, font=\normalsize] at (axis cs:6.75,6) {Restoration};
		
		\draw[|<->|, thick] (axis cs:8,6) -- (axis cs:10,6);
		\node[fill=white, font=\normalsize] at (axis cs:9,6) {Normal State};
	\end{axis}
\end{tikzpicture}
	\vspace*{-.5em}
	\caption{%
		General resilience model with four different phases.
		For the system considered in this work, the performance corresponds to the \gls{sk} budget.
		In the normal state, messages are transmitted with probability~$\probtx$ while \gls{skg} is performed the rest of the time.
		During the alert state, a message is transmitted in each time slot, i.e., $\probtx=1$.
		The duration of the alert state~$\alertduration$ is random but follows a known distribution.
	}%
	\label{fig:resilience-model}
\end{figure*}

In this work, we are interested in the \emph{resilience} of the communication system with a secret-key budget described above.
Throughout the following, the basic model for describing the resilience of a general system depicted in \autoref{fig:resilience-model} is used.
This concept is popular in other research areas, such as in power systems and smart grid~\cite{Panteli2015,Moreno2020,Stankovi2023,Ji2017,Xu2024}.
The operation of the system is categorized into the following four distinct states.

\subsubsection{Normal State}
During the normal operation state, the system functions as specified, and its parameters stay within predefined ranges.
It is important to note that even in this state, there might be fluctuations of the performance.

In the system considered in this work, this state is described above.
In each time slot, a message is transmitted with probability~$\probtx$, while \gls{skg} is performed otherwise.
The performance of the system corresponds to the number of available \gls{sk} bits~$\budget$.

\subsubsection{Alert State}
Besides the normal state, the system can also enter an alert state, which is typically triggered by external events like extreme weather conditions or active attacks.
In the alert state, the system parameters are outside of the normal ranges and the overall performance of the system degrades.

For this work, we assume that the system could enter an externally triggered alert state at any time.
In this state, a message must be transmitted in every time slot, i.e., $\probtx=1$.
The duration of this alert state is denoted as~$\alertduration$, which we assume to be a random variable with a known distribution.

\subsubsection{Degraded Performance}
After the alert state ends, the system's performance might remain in a degraded state.
However, the system should be designed in a way that ensures basic functionality during this period.

In this work, we assume that the system immediately switches to a restoration mode once the alert state is over.
Therefore, the duration of the state of degraded performance is assumed to be zero.

\subsubsection{Restoration}
Once the external circumstances permit restoration of the system, the performance gradually returns to the normal operation state.
At this point, all parameters are back within normal ranges and the system operates as specified again.

While the restoration phase is an important part of the resilience cycle, we do not focus on it in this work due to space limitations.
However, it will be important to consider different strategies for this phase in future work.

\subsection{Resilience Metric}\label{sub:resilience-metric}
In order to quantify the resilience, we first need to introduce some new resilience metrics.
In general, we define the system to fail if it runs out of key bits.
In this work, we are particularly interested in preparing the system for an eventual alert state, i.e., we consider the system to be in the normal state of operation and aim to define appropriate resilience metrics that quantify the preparedness if an alert state would occur in the next time slot.
Given that the system's performance is inherently stochastic, it is also necessary to assess the resilience in terms of probabilities.
In particular, we introduce two new quantities in the following.

\subsubsection{Alert Outage Probability}
The first important aspect of resilience is surviving alert states.
As the duration of the alert state~$\alertduration$ is random, we first introduce the \emph{alert survival probability~$1-\alertoutprob(t)$} as the probability that the system does not run out of \gls{sk} bits when entering the alert state in the next time slot.
Furthermore, we require the system to also survive until time slot~$t$, i.e., it has not run out of key bits up to this point.
From this, we get the definition of~$1-\alertoutprob(t)$ as
\begin{equation}\label{eq:def-alert-outage-prob}
	1-\alertoutprob(t) = \Pr\left(\budget(t) > \sum_{i=t}^{t+\alertduration}L, \min_{0 \leq i \leq t}\budget(i) > 0\right).
\end{equation}
Similarly, we refer to $\alertoutprob(t)$ as the \emph{alert outage probability}.

\subsubsection{Resilience Outage Probability}
Throughout the following, we assume that our application sets a specific target for the alert outage probability~$\alertoutprobmax$.
This means that the system should be designed such that it always has enough \gls{sk} bits during the normal state to only run out of key bits with a probability of at most $\alertoutprobmax$ when entering an alert state.
Equivalently, this means that the system should have enough key bits to survive ${(1-\alertoutprobmax)}$ of the alert states.
Combining the duration of the alert states with the message length yields the amount of bits that should always be available as
\begin{equation}\label{eq:def-min-budget}
	\minbudget = \inv{\cdf_{\alertduration}}\left(1-\alertoutprobmax\right)\cdot\messagelength\,,
\end{equation}
where $\inv{\cdf_{\alertduration}}$ is the quantile function of the alert state duration~$\alertduration$.

Every time the budget~$\budget$ falls below the threshold~$\minbudget$ is considered a violation of this resilience requirement.
We will refer to this event as a resilience outage with the corresponding \emph{resilience outage probability~$\resiloutageprob$}
\begin{equation}\label{eq:def-resilience-outage-prob}
	\resiloutageprob(t) = \Pr\left(\alertoutprob(t) > \alertoutprobmax\right)\,.
\end{equation}

\begin{example}[Resilience Metrics]\label{ex:resilience-metric}
In order to illustrate the introduced resilience metrics, we use the following numerical example.
The system starts with an initial budget~$\budget(0)=\initbudget=\SI{20}{\bit}$ at time~${t=0}$.
With probability~$\probtx=0.25$, a message of length~$\messagelength=\SI{2}{\bit}$ arrives in a time slot, reducing the \gls{sk} budget~$\budget$ by~$L$.
If no message arrives, \gls{skg} is performed and the number of available bits increases.
An exemplary progression of the budget~$\budget(t)$ over time is depicted in \autoref{fig:illustration-resilience-metrics}.

The duration of an alert state~$\alertduration$ is assumed to be distributed according to a Poisson distribution with mean~\num{6}, i.e., ${\alertduration\sim\poisdist(6)}$.
Furthermore, it is specified for the application that it should be able to survive an alert state with a probability of at least~\SI{90}{\percent}, i.e., $\alertoutprobmax=0.1$.
According to \eqref{eq:def-min-budget}, this translates to the target that the number of available key bits~$\budget(t)$ should always stay above~$\minbudget=L\inv{F_{\alertduration}}(1-\alertoutprobmax)=\SI{18}{\bit}$.
Therefore, every time the system's budget drops below that threshold~$\minbudget$, there is a resilience outage.
These events are highlighted by a red background in \autoref{fig:illustration-resilience-metrics}.
For the depicted example, the systems violates the resilience requirement in time slots~\num{17}, \num{18}, and \num{25} to \num{30}.
The probability that such a resilience outage occurs in time slot~$t$ is the resilience outage probability~$\resiloutageprob(t)$ from~\eqref{eq:def-resilience-outage-prob}.

\begin{figure}
	\centering
	\begin{tikzpicture}
\pgfplotsset{
	scale only axis,
	xmin=0,
	xmax=30,
}

\pgfplotstableread[x=time]{data/samples-illustration-resilience-metrics.dat}\tbl

\begin{axis}[
	betterplot,
	width=.68\linewidth,
	height=.19\textheight,
	axis y line*=left,
	ymin=0,
	ymax=30,
	xlabel={Time~$t$},
	ylabel={Budget~$\budget(t)$ [\si{\bit}]},
	extra y ticks={18},
	extra y tick labels={$\minbudget$},
	]
	\addplot+ table[y={budget}] {\tbl};
	\label{plot_budget}
	
	\addplot[plot0, dashed, very thick, domain=0:30, samples=2] {18};
\end{axis}

\begin{axis}[
	width=.68\linewidth,
	height=.19\textheight,
	cycle list name=lineplot cycle,
	axis y line*=right,
	axis x line=none,
	ymin=1e-3,
	ymax=1,
	ymode=log,
	ylabel={Alert Outage Probability~$\alertoutprob(t)$},
	legend pos=south west,
	legend cell align=left,
	ytick={1e-3, 1e-2, 1e-1, 1},
	yticklabels={$10^{-3}$, $10^{-2}$, $\alertoutprobmax$, $10^{0}$},
	ylabel shift=-.2cm,
	]
	\addlegendimage{/pgfplots/refstyle=plot_budget}\addlegendentry{Budget~$\budget$}
	
	\pgfplotsset{cycle list shift=1};
	\addplot+ table[y=alert] {\tbl};
	\addlegendentry{Alert Outage Probability~$\alertoutprob$}
	
	\addplot[plot1, dashed, very thick, domain=0:30, samples=2] {1e-1};
	
	\draw[plot2,fill=plot2,opacity=.2] (16.5, 1e-3) rectangle (18.5, 1);
	\draw[plot2,fill=plot2,opacity=.2] (25, 1e-3) rectangle (30, 1);
\end{axis}
\end{tikzpicture}
	\vspace*{-.5em}
	\caption{%
		Illustration of the relation between \gls{sk} budget, alert outage probability~$\alertoutprob$, and resilience outage events.
		The target maximum alert outage probability is set to $\alertoutprobmax=10^{-1}$.
		The highlighted areas indicate the time slots~$t$ at which the alert outage probability exceeds the threshold, i.e., $\alertoutprob(t)>\alertoutprobmax$.
		(\autoref{ex:resilience-metric})}%
	\label{fig:illustration-resilience-metrics}
\end{figure}
\end{example}

With these newly introduced resilience metrics, we can now summarize the problem formulation for the remainder of this work.
\begin{prob*}
The aim of this work is to analyze the physical layer resilience of the described communication system with a secret-key budget.
In particular, we are interested in simultaneously optimizing the power consumption and resilience outage probability.
\end{prob*}

\subsection{Application Example}\label{sub:application-examples}
After introducing the system model along with the resilience model and appropriate new resilience metrics, we now take a closer look at a specific application example.

Consider a wireless sensor system, in which the transmitter is a small temperature sensor that securely reports the current temperature occasionally (with probability~$\probtx\ll 1$) to a fusion center.
For securing the messages, the system uses \gls{skg} and one-time pads as described above.
The underlying protocol for transmitting the messages could be any standardized transmission protocol with a fixed packet length~$\messagelength$, where each packet contains the current temperature reading as payload.
This corresponds to the normal state of operation. %

If the temperature rises above a certain threshold, it indicates a malfunction of the monitored system and it is crucial for the fusion center to get frequent updates of the temperature readings in this situation.
This state corresponds to the alert state, which is triggered by the temperature crossing a predefined threshold.
It should be noted that this trigger is external and not related to quantities within the communication chain.

In this work, we assume that the update frequency during the alert state is maximal, i.e., a temperature reading is transmitted in each time slot, which corresponds to ${\probtx=1}$.
Once the temperature readings drop below another predefined threshold again, the alert state is considered to be over, and the total duration between the crossings of the thresholds is the length of the alert state~$\alertduration$.

\begin{figure}
	\centering
	\begin{tikzpicture}[
		/pgfplots/scale only axis,
	]
	\begin{axis}[
		betterplot,
		width=.8\linewidth,
		height=.13\textheight,
		xlabel={Time~$t$},
		ylabel={Temperature},
		xmin=0,
		xmax=25,
		ymin=0,
		ymax=1,
		domain=0:25,
		yticklabel=\empty,
		name=temperature axis,
		extra y ticks={.6,.8},
		extra y tick labels={$T_0$, $T_1$},
		]
		\addplot+[black] table[x=time,y=temperature] {data/samples_temperature.dat};

		\addplot[very thick,dashed,black] {.8};
		\addplot[very thick,dashed,black] {.6};

		\addplot[thin,fill=plot1,fill opacity=.15] coordinates {(10, 1.1) (10, -.1) (18, -.1) (18, 1.1)};
		\addplot[thin,fill=plot0,fill opacity=.15] coordinates {(10, 1.1) (10, -.1) (0, -.1) (0, 1.1)};
		\addplot[thin,fill=plot0,fill opacity=.15] coordinates {(25, 1.1) (25, -.1) (18, -.1) (18, 1.1)};

		\draw[|<->|, thick] (axis cs:0,.15) -- (axis cs:10,.15);
		\node[fill=white, font=\small] at (axis cs:5,.15) {Normal State};

		\draw[|<->|, thick] (axis cs:10,.15) -- (axis cs:18,.15);
		\node[fill=white, font=\small] at (axis cs:14,.15) {Alert State};

		\draw[|<->|, thick] (axis cs:18,.15) -- (axis cs:25,.15);
		\node[fill=white, font=\small] at (axis cs:21.5,.15) {Normal State};
	\end{axis}

	\begin{axis}[
		anchor=south west,
		at=(temperature axis.north west),
		width=.8\linewidth,
		height=.02\textheight,
		xtick=\empty,
		ytick=\empty,
		xmin=0,
		xmax=25,
		ymin=0,
		ymax=1,
		]
		\draw[plot3,fill=plot3,fill opacity=.3] foreach \n in {2, 6, 8}{(axis cs:\n,-1) rectangle (axis cs:\n+1,1.1)};
		\draw[plot3,fill=plot3,fill opacity=.3] foreach \n in {10, ..., 17}{(axis cs:\n,-1) rectangle (axis cs:\n+1,1.1)};
		\draw[plot3,fill=plot3,fill opacity=.3] foreach \n in {20, 22, 23}{(axis cs:\n,-1) rectangle (axis cs:\n+1,1.1)};

		\draw[plot2,fill=plot2,fill opacity=.3] foreach \n in {0, 1, 3, 4, 5, 7, 9}{(axis cs:\n,-1) rectangle (axis cs:\n+1,1.1)};
		\draw[plot2,fill=plot2,fill opacity=.3] foreach \n in {18, 19, 21, 24, 25}{(axis cs:\n,-1) rectangle (axis cs:\n+1,1.1)};
	\end{axis}
\end{tikzpicture}
	\caption{Illustration of an application example in which a temperature sensor measures the temperature and transmits its reading as a payload with probability~$\probtx$.
		If a critical temperature threshold~$T_1$ is crossed, the system enters an alert state until the temperature falls below a second threshold~$T_0$ again.
		During the alert state, a temperature reading is transmitted in every time slot.
		During the normal state, a reading is transmitted with probability~$\probtx$ (indicated by the black boxes above the time slots).
		If no message is transmitted, \gls{skg} is performed instead (indicated by the red boxes above the time slots).}
	\label{fig:temperature-example}
\end{figure}

An illustration of this application example is given in \autoref{fig:temperature-example}.
Once the temperature reading crosses the upper threshold, the system switches to the alert state, which causes messages to be transmitted in every time slot (indicated by the black boxes above).
During the normal state of operation, \gls{skg} is performed if no message needs to be transmitted (indicated by the red boxes).

Beyond this example, potential applications of the model and the results presented in this work include other communication systems with a budget, e.g., battery-constrained systems.
\section{Power Control Schemes for Resilience}\label{sec:resilience-analysis}

Since we are considering the normal operation state of the resilience cycle in this work, we focus on the long-term behavior of the system and its preparation for an alert state.
In this section, we present three different power allocation schemes.

\subsection{Constant Power}\label{sub:const-power}
As a first basic scheme, we consider a constant transmission power.
In this case, the transmit power is fixed to a constant value for all \gls{skg} time slots, i.e., we have ${\txpower(t)=\txpower}$, for all time slots~$t$.
In the following, we derive analytical performance bounds and consider the long-term behavior of the system.

A first result are lower bound and upper bound on the resilience outage probability~$\resiloutageprob$.
These bounds are easier to calculate than the direct expression from~\eqref{eq:def-resilience-outage-prob}.

\begin{thm}[{Bounds on the Resilience Outage Probability}]\label{thm:bounds-resilience-outage-prob}
Consider the described communication system with an \gls{sk} budget.
The resilience outage probability is lower bounded by
\begin{equation}\label{eq:lower-bound-resilience-outage-prob}
	1-\min\left\{\cdf_{\netsum(t)}\big(\initbudget-\minbudget\big),\; \probsurv_{\initbudget}(t)\right\} \leq \resiloutageprob(t)
\end{equation}
and upper bounded by
\begin{equation}\label{eq:upper-bound-resilience-outage-prob}
	\resiloutageprob(t) \leq
	1-\positive{\cdf_{\netsum(t)}\big(\initbudget-\minbudget\big) - \probruin_{\initbudget}(t)}\,,
\end{equation}
with the ruin probability~$\probruin_{\initbudget}$ defined in~\eqref{eq:def-prob-ruin}.
\end{thm}
\begin{proof}
	The proof can be found in \autoref{app:proof-thm-bounds-resilience-outage-prob}.
\end{proof}

Ideally, the system should operate indefinitely under normal conditions.
Consequently, an important aspect is the long-term behavior of the resilience outage probability.
In particular, we are interested in adjusting the transmit power level such that the resilience outage probability~$\resiloutageprob$ does not exceed a specified threshold in the steady state of the system.
In order to analyze this, we rely on the following result established in~\cite{Besser2024SKGbudget}.
\begin{lem}[{\cite[Cor.~2]{Besser2024SKGbudget}}]\label{lem:expected-net-claim-random}
	Consider the described communication system in the {normal state} where a message is transmitted with probability~$\probtx$.
	The following relation between the expected value of the net usage~$\netusage$ and the transmission probability~$\probtx$ holds:
	\begin{equation}\label{eq:relation-tx-prob-neg-expect-net-claim-random}
		\expect{\netusage}\gtreqqless 0 \quad \Leftrightarrow \quad \probtx \gtreqqless \frac{\expect{\rateskg}}{\expect{\rateskg} + L} = \probtxcrit\,.
	\end{equation}
\end{lem}

An important consequence for the resilience outage probability, following from \autoref{lem:expected-net-claim-random}, is given in the following theorem.
\begin{thm}[{Long-Term Resilience Outage Probability}]\label{thm:long-term-resilience-outage}
	Consider the described communication system in the {normal state} where a message is transmitted with probability~$\probtx$.
	For $\probtx>\probtxcrit$, the system will run out of \gls{sk} bits almost surely, and the long-term resilience outage probability is therefore
	\begin{equation}
		\lim\limits_{t\to\infty} \resiloutageprob(t) = 1\,.
	\end{equation}
	For $\probtx<\probtxcrit$, we have
	\begin{equation}\label{eq:resil-outage-prob-long-term-small-p}
		\lim\limits_{t\to\infty}\resiloutageprob(t) = \lim\limits_{t\to\infty}\probruin_{\initbudget}(t) = \probruin_{\initbudget}(\infty)\,.
	\end{equation}
\end{thm}
\begin{proof}
	The proof can be found in \autoref{app:proof-cor-long-term-resilience-outage}.
\end{proof}

Combining the above results allows us to answer the question from our problem formulation regarding the trade-off between transmit power and resilience for a constant power allocation.
In particular, we can now determine the minimum (constant) transmit power such that the resilience requirements are still fulfilled.
Given the system parameters and resilience requirements, this corresponds to the smallest power~$\txpower$ such that the probability of ultimate ruin~$\probruin_{\initbudget}(\infty)$ is equal to the accepted resilience outage probability for the application.
This will be illustrated in the following example.

\begin{example}[Constant Power -- Rayleigh Fading]\label{ex:constant-power}
In the following, we illustrate the general results from above with a numerical example.
In particular, we assume that the channels between Alice and Bob, and Alice and Eve follow Rayleigh fading with average channel \glspl{snr}~$\expect{\channelbob}=\SI{10}{\dB}$ and $\expect{\channeleve}=\SI{0}{\dB}$.
The other system parameters are set to $\initbudget=\SI{70}{\bit}$, $\probtx=0.35$, $\messagelength=\SI{5}{\bit}$ and $\alertoutprobmax=10^{-1}$.
The duration of an alert state is assumed to be Poisson-distributed with an average length of \num{5}~time slots, i.e., $\alertduration\sim\poisdist(5)$.
According to~\eqref{eq:def-min-budget}, this implies that the system needs an \gls{sk} budget of at least $\minbudget=\SI{40}{\bit}$ to fulfill the target alert outage probability.

The behavior of the resulting resilience outage probability~$\resiloutageprob(t)$, ruin probability~$\probruin_{\initbudget}(t)$, and probability of ultimate ruin~$\probruin_{\initbudget}(\infty)$ over time are shown in \autoref{fig:outage-prob-vs-time}.
The transmit power is set to $\txpower=\SI{10}{\dB}$.
The resilience outage probability is determined by \gls{mc} simulations with $10^{6}$~samples.
The ruin probability in finite time~$\probruin_{\initbudget}(t)$ is determined according to the integrodifference equation~\cite{Lutscher2019} given in~\cite[{Eq.~(12)}]{Besser2024SKGbudget}, and the probability of ultimate ruin is calculated according to the integral equation from~\cite[{Eq.~(24)}]{Besser2024SKGbudget}.
The code to reproduce all of the shown results can be found in~\cite{BesserGithub}.

\begin{figure}
	\centering
	\begin{tikzpicture}
	\begin{axis}[
		betterplot,
		xlabel={Time~$t$},
		ylabel={Outage Probabilities},
		xmin=0,
		xmax=250,
		ymode=log,
		ymin=1e-4,
		ymax=1,
		legend pos=south east,
		legend style={%
			font=\normalsize,
		},
		]
		\pgfplotstableread[x=time]{data/results-p0.35-P10.0-b70.0-e1.000000E-01.dat}\tbl

		\addplot+[mark repeat=20] table[y={outageMC}] {\tbl};
		\addlegendentry{Resil. Outage Prob.~$\resiloutageprob(t)$}
		\addplot+[mark repeat=20] table[y={ruinIDE}] {\tbl};
		\addlegendentry{Ruin Prob.~$\probruin_{\initbudget}(t)$}
		\addplot+[mark repeat=20] table[y={ultimateRuin}] {\tbl};
		\addlegendentry{Prob. of Ult. Ruin~$\probruin_{\initbudget}(\infty)$}
	\end{axis}
\end{tikzpicture}
	\caption{Behavior of different outage probabilities over time for a system with parameters $\expect{\channelbob}=\SI{10}{\dB}$, $\expect{\channeleve}=\SI{0}{\dB}$, $\txpower=\SI{10}{\dB}$, $\probtx=0.35$, $\messagelength=\SI{5}{\bit}$, $\initbudget=\SI{70}{\bit}$, and $\minbudget=\SI{40}{\bit}$.
	(\autoref{ex:constant-power})}%
	\label{fig:outage-prob-vs-time}%
\end{figure}%

For the chosen system parameters with $\txpower=\SI{10}{\dB}$, the critical probability~$\probtxcrit$ at which the average usage of \gls{sk} bits becomes negative is calculated according to~\eqref{eq:relation-tx-prob-neg-expect-net-claim-random} to ${\probtxcrit=0.398}$.
The average usage of key bits is $\expect{\netusage}=\SI{-0.4}{\bit}$.
Thus, the resilience outage probability~$\resiloutageprob(t)$ will approach the probability of ultimate ruin~$\probruin_{\initbudget}(\infty)$ over time according to \autoref{thm:long-term-resilience-outage}.
This is calculated to $\probruin_{\initbudget}(\infty)=0.044$ for the above values.
First, it can be seen from \autoref{fig:outage-prob-vs-time} that the ruin probability~$\probruin_{\initbudget}(t)$ increases over time and approaches $\probruin_{\initbudget}(\infty)$.
Similarly, the resilience outage probability~$\resiloutageprob(t)$ converges also to $\probruin_{\initbudget}(\infty)$, which is consistent with the result from \autoref{thm:long-term-resilience-outage}.
The operational meaning of this result is that in about $\SI{4.4}{\percent}$ of the cases, the system has less \gls{sk} bits available than it would need to meet the resilience requirement~$\alertoutprobmax$.

Next, we show the influence of the transmit power~$\txpower$ on the resilience outage probability in \autoref{fig:outage-prob-vs-power}.
The resilience outage probabilities~$\resiloutageprob(t)$ and finite-times ruin probabilities~$\probruin_{\initbudget}(t)$ are shown for time slot $t=200$.
It can be seen that the gap between them is small for all transmit powers.
In contrast, the gap between $\resiloutageprob(t)$ and the probability of ultimate ruin~$\probruin_{\initbudget}(\infty)$ is larger for small $\txpower$.
This indicates that the convergence is slower for small transmit powers.
Additionally, it can be seen that $\probruin_{\initbudget}(\infty)=1$ for small $\txpower$.
This can be found directly by calculating the critical transmission probabilities~$\probtxcrit$, which are below the system's $\probtx=0.35$ for small $\txpower$, e.g., for $\txpower=\SI{2}{\dB}$ we have $\probtxcrit=0.338$.
Inversely, based on \eqref{eq:relation-tx-prob-neg-expect-net-claim-random}, we can calculate the maximum transmit power up to which the probability of ultimate ruin~$\probruin_{\initbudget}(\infty)$ is one.
For the given parameters, this is around $\txpower=\SI{3.12}{\dB}$.

A system designer could now also directly determine the minimum transmit power to fulfill the resilience requirements of the system.
Assuming that the application tolerates a maximum resilience outage probability of $\resiloutageprobmax=10^{-1}$, the minimum transmit power would be around $\txpower=\SI{7.5}{\dB}$, cf.~\autoref{fig:outage-prob-vs-power}.

\begin{figure}
	\centering
	\begin{tikzpicture}
	\begin{axis}[
		betterplot,
		xlabel={Transmit Power~$\txpower$ [$\si{\dB}$]},
		ylabel={Outage Probabilities},
		xmin=0,
		xmax=20,
		ymode=log,
		ymin=1e-3,
		ymax=1,
		legend pos=south west,
		extra x ticks={7.5},
		legend style={%
			font=\normalsize,
		},
		]
		\pgfplotstableread{data/combined_results-t200.000-p0.350-b70.000-e0.100-mod.dat}\tbb

		\addplot+ table[x={power_tx_db},y={outageMC}] {\tbb};
		\addlegendentry{Resil. Outage Prob.~$\resiloutageprob({200})$}
		\addplot+ table[x={power_tx_db},y={ruinIDE}] {\tbb};
		\addlegendentry{Ruin Prob.~$\probruin_{\initbudget}({200})$}
		\addplot+ table[x={power_tx_db},y={ultimateRuin}] {\tbb};
		\addlegendentry{Prob. of Ult. Ruin~$\probruin_{\initbudget}(\infty)$}

		\addplot[black, dashed, very thick] coordinates {(0,0.1) (7.5,.1) (7.5,1e-4)};
	\end{axis}
\end{tikzpicture}
	\caption{Influence of the transmit power~$\txpower$ on different outage probabilities for a system with parameters $\expect{\channelbob}=\SI{10}{\dB}$, $\expect{\channeleve}=\SI{0}{\dB}$, $\probtx=0.35$, $\messagelength=\SI{5}{\bit}$, $\initbudget=\SI{70}{\bit}$, and $\minbudget=\SI{40}{\bit}$.
	For a maximum tolerated resilience outage probability of~${\resiloutageprobmax=10^{-1}}$, the transmit power could be lowered to around~$\txpower=\SI{7.5}{\dB}$.
	(\autoref{ex:constant-power})
	}%
	\label{fig:outage-prob-vs-power}
\end{figure}
\end{example}

\subsection{Adaptive Power Control}\label{sub:adaptive}
Since we assume \gls{csi} at the communication parties, this information could be leveraged to improve the performance, i.e., reduce the transmit power whenever possible.
In the following, we propose an adaptive power allocation scheme based on an adapted \gls{ee} notion, which uses the conditional expected key generation rate.
In particular, we can assume that the channel gain between Alice and Bob~$\channelbob(t)$ is known accurately.
Additionally, we assume statistical \gls{csi} of the eavesdropper's channel, i.e., the distribution of~$\channeleve$ is known at the transmitter.
With this information, we design the following adaptive power allocation strategy.

\begin{prop}[{Adaptive Power Control}]\label{prop:adaptive-power-control}
	Consider the communication system with a \gls{sk} budget as described above.
	At time~$t$, the transmit power~$\txpower(t)$ is given as
	\begin{equation}\label{eq:adaptive-power-control}
		\txpower(t) =
		\begin{cases}
			\txpowermax & \text{if}\ \budget(t)-\minbudget \leq 0\\
			\argmax_{P} \adaptee(P) & \text{otherwise}
		\end{cases}
	\end{equation}
	with the adapted \gls{ee}
	\begin{equation}\label{eq:def-adapt-ee}
		\adaptee(P) = \frac{\expect{\rateskg(P) \middle| \channelbob(t)}}{P^{\weight(\budget(t)-\minbudget)}}\,,
	\end{equation}
	where $\weight > 0$ is a hyperparameter that allows tuning the importance of reducing the transmit power, with larger~$\weight$ enforcing lower transmit powers.
\end{prop}

The primary goal is to find a power control strategy that improves the resilience in terms of the resilience outage probability~$\resiloutageprob$, i.e., the budget of the system~$\budget(t)$ should stay above the required minimum~$\minbudget$ as much as possible.
The adaptive power control in \autoref{prop:adaptive-power-control} combines this target with the second objective of minimizing the transmit power using the following ideas:
\begin{enumerate}
	\item When the system is currently in a resilience outage, i.e., $\budget(t) \leq \minbudget$, the system needs to recover as quickly as possible.
	This is done by allocating the maximum allowed power~$\txpowermax$ as it maximizes the \gls{skg} rate.
	\item However, if the system is not in outage, we have the chance to save power.
	In particular, we maximize the \gls{ee}, which is generally given as the ratio of achievable data rate and consumed power.
	In our case, the data rate is the expected \gls{skg} rate given the channel gain of the main channel~$\channelbob(t)$.
	Additionally, we add an exponent~${\weight(\budget(t)-\minbudget)}$, with parameter~$\weight > 0$, to the transmit power that allows a flexible tuning of the importance of the two conflicting objectives (maximum rate vs.\ minimal power).
	A larger value of~${\weight(\budget(t)-\minbudget)}$ increases the importance of saving energy and a lower transmit power will be chosen according to~\eqref{eq:adaptive-power-control}.
	Additionally, incorporating the difference of the current budget~$\budget(t)$ and the target minimum budget~$\minbudget$ into the exponent, automatically increases the importance of achieving a high \gls{skg} rate over saving energy when the system gets close to a resilience outage.
\end{enumerate}

\begin{example}[{Adaptive Power Control -- Rayleigh Fading}]\label{ex:adaptive-rayleigh}
In order to illustrate the adaptive power control scheme, we again use Rayleigh fading as a numerical example.
We use the system parameters from \autoref{ex:constant-power}, i.e., $\expect{\channelbob}=\SI{10}{\dB}$ and $\expect{\channeleve} = 1/\ly = \SI{0}{\dB}$.
Using the expression of the \gls{skg} rate from~\eqref{eq:skg-rate} and the exponential distribution of the channel gains, the conditional expectation is calculated as
\begin{multline*}
	\expect{\rateskg(P) \middle| \channelbob=h} = \\
	\int_{0}^{\infty} \log_2\left(\frac{1 + P h + P y}{1 + P y}\right) \ly \exp(-\ly y) \diff{y}\,,
\end{multline*}
which is evaluated to~\eqref{eq:cond-expect-skg-rayleigh} {at the bottom of this page}
with $\Ei$ being the exponential integral~\cite[Sec.~5.1]{Abramowitz1972}.

\capstartfalse
\begin{table*}[b]
	\hrulefill
	\normalsize
	\begin{equation}\label{eq:cond-expect-skg-rayleigh}
		\expect{\rateskg(P) \middle| \channelbob=h} = \frac{1}{\log 2} \Bigg(\exp\left(\frac{\ly}{P}\right) \Ei\left(-\frac{\ly}{P}\right)
		- \exp\left(\ly \left(h + \frac{1}{P}\right)\right) \Ei\left(-\frac{\ly + h \ly P}{P}\right) + \log(1 + h P)\Bigg)
	\end{equation}
\end{table*}
\capstarttrue

A visualization of the adapted \gls{ee}~$\adaptee$ from~\eqref{eq:def-adapt-ee} is shown for different values of $\channelbob$ and $\budget(t)-\minbudget$ in \autoref{fig:adaptive-target-function-rayleigh}.
The two primary ideas of the adaptive power control scheme can both be seen.
First, for small budget surpluses, e.g., $\budget-\minbudget=10$, the maximum of~$\adaptee$ is attained at higher transmit powers.
This ensures a higher \gls{skg} rate and reduces the likelihood of a resilience outage.
Correspondingly, at a high budget surplus, e.g., $\budget-\minbudget=50$, energy can be conserved and the transmit power~${\argmax g}$ according to~\eqref{eq:adaptive-power-control} is lower.
The second shown effect is the exploitation of the available \gls{csi} of the main channel.
When the channel to Bob is good, i.e., $\channelbob(t)$ is high, the transmit power can be reduced while still achieving a high \gls{skg} rate.
In contrast, when $\channelbob(t)$ is small, more power needs to be used.

\begin{figure}[t]%
	\centering
	\begin{tikzpicture}%
	\begin{axis}[
		betterplot,
		xlabel={Transmit Power~$P$},
		ylabel={Adapted Energy Efficiency~$\adaptee(P)$},
		xmin=1,
		xmax=1000,
		xmode=log,
		ymin=0,
		ymax=5,
		legend style={%
			font=\small,
			anchor=east,
			at={(.98, .53)},
		}
		]
		
		\pgfplotstableread[x=power,col sep=tab]{data/adaptive-target-function-ly1.0.dat}\tbl
		
		\addplot+[mark repeat=10] table[x=power,y=hb1b10] {\tbl};
		\addlegendentry{$\channelbob=1$, $\budget-\minbudget=10$}
		
		\addplot+[mark repeat=10] table[x=power,y=hb1b50] {\tbl};
		\addlegendentry{$\channelbob=1$, $\budget-\minbudget=50$}
		
		\addplot+[mark repeat=10] table[x=power,y=hb20b10] {\tbl};
		\addlegendentry{$\channelbob=20$, $\budget-\minbudget=10$}
		
		\addplot+[mark repeat=10] table[x=power,y=hb20b50] {\tbl};
		\addlegendentry{$\channelbob=20$, $\budget-\minbudget=50$}
	\end{axis}
\end{tikzpicture}
	\vspace*{-.5em}
	\caption{Adapted \gls{ee}~$\adaptee$ from~\eqref{eq:def-adapt-ee} for Rayleigh fading with parameters $\ly=1$ and $\weight=0.002$. (\autoref{ex:adaptive-rayleigh})}%
	\label{fig:adaptive-target-function-rayleigh}
\end{figure}

\end{example}

\subsection{Reinforcement Learning}\label{sub:reinforcement}
Due to the multi-objective and dynamic nature of the optimization problem, it is difficult to find an optimal power allocation strategy.
While the adaptive algorithm from \autoref{sub:adaptive} takes the current state of the system into account when calculating the power level, it does not incorporate the history and dynamic nature of the system.
A good decision strategy for such dynamic time series optimization problems can be found by \gls{ml} algorithms, in particular \gls{rl}~\cite{Sutton2020}.

The basic setup of \gls{rl} is an agent in a dynamic environment.
The agent observes the current state of the environment and takes an action based on it.
This action changes the state of the system and the agent receives a reward depending on how useful the action was.
The goal is to train the agent's action policy such that the cumulative reward is maximized.

In this work, we use \gls{rl} to find a power allocation strategy for improving the resilience of communication systems with a \gls{sk} budget.
However, it should be noted that the primary goal in this work is \emph{not} to optimize the \gls{rl}-based solution but to provide one way of leveraging \gls{ml} in the context of \gls{sk} budgets.
The investigation of other reward functions and network structures is outside the scope of this work.

\subsubsection{Action and Observation Space}
The first important aspect of defining the \gls{rl} system is specifying the observation and action space of the agent.
Since the task is to find a transmission power, the action is a continuous number in the interval~${[0, \txpowermax]}$.
The observation space corresponds to the information that the agent can acquire before making its decision.
In our scenario, this includes the current budget~$\budget(t)$ and \gls{csi} about the main channel~$\channelbob(t)$.
Additionally, the system knows whether a message is to be transmitted, i.e., it has access to~$\txindicator(t)$.

\subsubsection{Reward Function}
The next central aspect of a successful application of \gls{rl} for the power allocation problem is a good formulation of the reward function.
In the considered problem, multiple conflicting objectives (maximal resilience vs. minimal power consumption) need to be balanced.
For this, we propose the following reward function~$\reward$ for time slots~$t$ in which \gls{skg} is performed, i.e., $\txindicator(t)=0$,
\begin{equation}\label{eq:reward-function}
	\reward(t) = \weight_1\frac{\avg{\txpower}(t)}{\txpowermax} + \weight_2 \log_{10}\resiloutageprob + \weight_3 \log_{10}\alertoutprob(t) + \weight_4 \! \negative{1-\frac{\minbudget}{\budget(t)}}
\end{equation}
with hyperparameters~${\weight_1 < 0}$, ${\weight_2 < 0}$, ${\weight_3 < 0}$ and ${\weight_4 > 0}$, and the average transmit power~$\avg{\txpower}(t)$ until time~$t$
\begin{equation*}
	\avg{\txpower}(t) = \frac{1}{t} \sum_{i=1}^{t} \txpower(i)\,.
\end{equation*}
The reward function~$\reward$ is the weighted sum of four individual parts.
The first part~$\avg{\txpower}/\txpowermax$ is the average transmit power normalized by the maximum allowed.
Since the corresponding weight~$\weight_1$ is negative, a high average power adds a penalty, incentivizing the agent to reduce the average power.
The value is normalized by the maximum to limit the range to ${[0, 1]}$ and, thus, make it more compatible with the other parts of the reward function.
While this covers the objective to reduce the consumed power, we use the remaining parts of~$\reward$ to model the resilience objective of the system.
The second and third part with corresponding weights~$\weight_2$ and $\weight_3$, respectively, add a reward for a low resilience outage probability~$\resiloutageprob$ and alert outage probability~$\alertoutprob$, respectively.
Decreasing outage probabilities are transformed into more negative values that are converted into higher rewards in combination with the negative weights~$\weight_2$ and $\weight_3$.
This adds the incentive for the agent to minimize both the instantaneous alert outage probability~$\alertoutprob(t)$ and the long-term resilience outage probability~$\resiloutageprob$.
Additionally, we have the fourth part which adds a large penalty whenever the current budget~$\budget(t)$ falls below the minimum required budget~$\minbudget$, i.e., whenever we have a resilience outage in the current time slot.
If the budget is sufficiently large, i.e., ${\budget(t) > \minbudget}$, this penalty is zero.
Since the full reward~$\reward$ is a weighted sum, changing the weights~$\weight_i$ allows a flexible tuning of the importance of the individual aspects.

In time slots in which a message is transmitted, i.e., ${\txindicator=1}$, the reward is set to a small positive constant value, e.g., ${\reward=1}$.
First, this highlights that the reward is independent of the action in these time slots, and additionally it encourages the agent to stay alive longer.

\subsubsection{Training Algorithm}
For training the proposed \gls{rl} agent, many different learning algorithms exist, e.g., \gls{sac}~\cite{Haarnoja2018sac} and \gls{ppo}~\cite{Schulman2017}.
Based on empirical data, we have achieved the best performance using \gls{ppo}.
However, the choice of the training algorithm is part of the hyperparameter optimization and could be considered in detail in future work.

In \autoref{sec:numerical-example}, we will show a numerical example with specific values for the weights~$\weight_i$ and a performance evaluation of the proposed \gls{rl} solution.

\subsection{Comparison of the Proposed Schemes}
The three different schemes proposed above have different advantages and disadvantages.
In the following, we will review and compare them in terms of complexity and adaptability.

As a first criterion, we compare the complexities when setting up the system and during operation.
If a constant power allocation is used, the complexity when setting up the system is minimal.
The system designer only needs to calculate the required power level according to the desired level of resilience as shown in \autoref{sub:const-power}.
The adaptive algorithm presented in \autoref{sub:adaptive} requires more preparation effort as we need an expression or efficient computation of the conditional expectation inside the function~$\adaptee$ from~\eqref{eq:def-adapt-ee} and its maximum in~\eqref{eq:adaptive-power-control}.
The highest complexity when preparing the system for operation has the \gls{rl}-based approach since it requires a full training phase of the underlying \gls{nn} agent.
However, once the training is complete, the complexity during operation is relatively low as it only requires a forward-pass through the trained \gls{nn}.
In contrast, the adaptive algorithm also requires some complex computations in every time slot during operation when calculating the adapted \gls{ee}~$\adaptee$ and its maximum.
The constant power allocation scheme is the simplest during operation as it requires no computation at all.

The next criterion for comparison is dynamic adaptability.
In particular, we refer to the ability of a power allocation scheme to adjust to the normal fluctuations of performance during the normal operation state.
As discussed in \autoref{sub:adaptive}, this adaptability is built into the adaptive algorithm, since its power allocation decision in each time slot depends on the current budget safety margin~${\budget(t)-\minbudget}$ and channel state~$\channelbob(t)$.
However, the influence of the current state of the system is fixed in~\eqref{eq:adaptive-power-control} and its strength is only determined by the single parameter~$\weight$.
A more flexible way of adapting the power based on the current (and past) states of the system can be achieved with an \gls{ml} solution.
Depending on the reward function and underlying structure, the \gls{rl}-based power allocation can model a complex behavior that can take current and past system states into account.
This comes at the cost of a high training effort and less explainability compared to the adaptive algorithm.
In contrast to the two aforementioned schemes, the constant power allocation does not provide any adaptability as it does not take any part of the system's current or past states into account.

A summary of the comparison of the different algorithms is given in \autoref{tab:comparison-schemes}.
An implementation of the algorithms to reproduce the results in this work is made publicly available at~\cite{BesserGithub}.

\begin{table}%
\renewcommand*{\arraystretch}{1.25}
\centering
\caption{Comparison of the Proposed Power Allocation Schemes}%
\label{tab:comparison-schemes}
\begin{tabularx}{.98\linewidth}{l|XXl}
	\toprule
	\diagbox[width=.35\linewidth]{Property}{Scheme} & Constant & Adaptive & \gls{rl}\\
	\midrule
	Preparation Complexity & Low & Medium & High\\
	Operation Complexity & Very Low & Medium & Low--Medium\\
	Adaptability & Very Low & Medium & High\\
	\bottomrule
\end{tabularx}
\end{table}
\section{Numerical Example}\label{sec:numerical-example}

In order to compare the performance of the proposed algorithms, we evaluate them using a numerical example in this section.
In particular, we assume the same system parameters as in \autoref{ex:constant-power} and \autoref{ex:adaptive-rayleigh}, i.e., Rayleigh fading with average channel \glspl{snr}~$\expect{\channelbob}=\SI{10}{\dB}$ and $\expect{\channeleve}=\SI{0}{\dB}$.
The other system parameters are set to $\initbudget=\SI{70}{\bit}$, $\probtx=0.35$, $\messagelength=\SI{5}{\bit}$ and $\alertoutprobmax=10^{-1}$.
The duration of an alert state is assumed to be Poisson-distributed with an average length of five time slots, i.e., $\alertduration\sim\poisdist(5)$.
The maximum allowed power is~$\txpowermax=\SI{30}{\dB}$.

In the following, we compare the following power allocation schemes:
\begin{itemize}
	\item \texttt{Model}:
	This represents the \gls{rl}-based scheme as explained in \autoref{sub:reinforcement}.
	The hyperparameters of the reward function are set to $\weight_1=-50$, $\weight_2=-10$, $\weight_3=-1$, $\weight_4=10$.
	\item \texttt{Max.\ Power}:
	As a baseline comparison, we use the scheme where the maximum available power~$\txpowermax$ is used in each time slot.
	This scheme achieves the best possible resilience performance, making it a lower bound for the resilience outage probability. However, it is important to note that this performance comes at the expense of low energy efficiency.
	\item \texttt{Adaptive}:
	For this, we use the adaptive algorithm introduced in \autoref{sub:adaptive} with parameter~$\weight=0.002$.
	\item \texttt{$\txpower=\SI{10}{\dB}$}:
	Constant power allocation with a transmit power of $\txpower=\SI{10}{\dB}$.
	\item \texttt{Const.\ Budget}:
	Constant power allocation with transmit power~$\txpower$ such that the equation~${\expect{\rateskg}=\messagelength\probtx/(1-\probtx)}$ holds.
	Based on~\eqref{eq:relation-tx-prob-neg-expect-net-claim-random}, this ensures an average net usage of zero key bits in each time slot.
	For the selected system parameters, this is evaluated to ${\txpower=\SI{3.12}{\dB}}$. %
\end{itemize}

We evaluate the performance of the above algorithms using \gls{mc} simulations, averaging \num{2000} runs.
The source code to reproduce the results can be found at~\cite{BesserGithub}.

\begin{figure}
	\centering
	\begin{tikzpicture}
	\begin{axis}[
		betterplot,
		xlabel={Time~$t$},
		ylabel={Transmit Power~$\txpower(t)$ [$\si{\dB}$]},
		xmin=0,
		xmax=2000,
		y filter/.code={\pgfmathparse{10*log10(\pgfmathresult)}},
		legend pos=north east,
		]
		\pgfplotstableread[x=time,col sep=tab]{data/2024-03-13-01-25-PPO-2/power-p0.35-B70-T2000.dat}\tbl

		\addplot+[mark repeat=80] table[x=time,y={model}] {\tbl};
		\addlegendentry{Model}

		\addplot+[mark repeat=80] table[x=time,y={fullPower}] {\tbl};
		\addlegendentry{Max. Power}

		\addplot+[mark repeat=80] table[x=time,y={adaptiveCondExpect}] {\tbl};
		\addlegendentry{Adaptive}

		\addplot+[mark repeat=80] table[x=time,y={constant10dB}] {\tbl};
		\addlegendentry{$\txpower=\SI{10}{\dB}$}

		\addplot+[mark repeat=80] table[x=time,y={constant}] {\tbl};
		\addlegendentry{Const. Budget}

	\end{axis}
\end{tikzpicture}
	\vspace{-.5em}
	\caption{%
		Transmit power~$\txpower$ over time for a communication system with the parameters described in \autoref{sec:numerical-example}.
		The maximum allowed power is~${\txpowermax=\SI{30}{\dB}}$.
		\texttt{Model} refers to the trained \gls{rl} model, \texttt{Adaptive} is the algorithm introduced in \autoref{sub:adaptive}, and the remaining curves correspond to constant powers at different levels.
		The shown values are averaged over \num{2000}~\gls{mc} runs.
	}%
	\label{fig:results-avg-power}
\end{figure}

First, we show the average transmit power~$\txpower(t)$ over time~$t$ in \autoref{fig:results-avg-power}.
For the schemes~\texttt{Max.\ Power}, \texttt{$\txpower=\SI{10}{\dB}$}, and \texttt{Const.\ Budget}, the power is constant over time.
In contrast, the two dynamic schemes~\texttt{Model} and \texttt{Adaptive}, show a varying power level over time.
Although they are based on different ideas (learned vs. analytical approach), the allocated power shows a similar behavior.
In both cases, the systems starts by allocating a larger transmit power to increase the \gls{sk} budget.
Over time, the used power reduces and stabilizes at a constant value.
For the trained \gls{rl}-model, this value is around ${\txpower=\SI{11.9}{\dB}}$, whereas the adaptive algorithm stabilizes at around ${\txpower=\SI{9}{\dB}}$.
However, it should be noted that this value can be adjusted by changing the parameter~$\weight$, which controls the tradeoff between rate and power consumption.
Similarly, the performance of the \gls{rl} solution will change depending on the hyperparameters.

The different power allocation strategies result in the average \gls{sk} budget shown in \autoref{fig:results-budget}.
As expected, the budget remains approximately constant at the initial budget level~$\initbudget=\SI{70}{\bit}$ for the \texttt{Const.~Budget} scheme.
Similarly, the budget approaches a stable value of around~\SI{153}{\bit} with the adaptive power allocation.
For the three other schemes with a larger average power consumption, the average budget increases over time.

\begin{figure}
	\centering
	\begin{tikzpicture}
	\begin{axis}[
		betterplot,
		width=.91\linewidth,
		height=.29\textheight,
		xlabel={Time~$t$},
		ylabel={Average Budget~$\expect{\budget}$ [$\si{\bit}$]},
		xmin=0,
		xmax=2000,
		ymin=0,
		ymax=1000,
		legend pos=north west,
		extra y ticks={40},
		extra y tick labels={$\minbudget$},
		]
		\pgfplotstableread[x=time]{data/2024-03-13-01-25-PPO-2/budget-p0.35-B70-T2000.dat}\tbl

		\addplot+[mark repeat=80] table[x=time,y={model}] {\tbl};
		\addlegendentry{Model}

		\addplot+[mark repeat=80] table[x=time,y={fullPower}] {\tbl};
		\addlegendentry{Max.\ Power}

		\addplot+[mark repeat=80] table[x=time,y={adaptiveCondExpect}] {\tbl};
		\addlegendentry{Adaptive}

		\addplot+[mark repeat=80] table[x=time,y={constant10dB}] {\tbl};
		\addlegendentry{$\txpower=\SI{10}{\dB}$}

		\addplot+[mark repeat=80] table[x=time,y={constant}] {\tbl};
		\addlegendentry{Const.\ Budget}

	\end{axis}
\end{tikzpicture}
	\vspace{-.5em}
	\caption{%
		Average \gls{sk} budget~$\budget$ over time for a communication system with the parameters described in \autoref{sec:numerical-example}.
		\texttt{Model} refers to the trained \gls{rl} model, \texttt{Adaptive} is the algorithm introduced in \autoref{sub:adaptive}, and the remaining curves correspond to constant powers at different levels.
		The shown values are averaged over \num{2000}~\gls{mc} runs.}%
	\label{fig:results-budget}
\end{figure}

\begin{figure}
	\centering
	\begin{tikzpicture}
	\begin{axis}[
		betterplot,
		xlabel={Time~$t$},
		ylabel={Resilience Outage Probability~$\resiloutageprob$},
		xmin=0,
		xmax=2000,
		ymode=log,
		ymin=1e-4,
		ymax=1,
		legend pos=south east,
		]
		\pgfplotstableread[x=time]{data/2024-03-13-01-25-PPO-2/resilience_outage_prob-p0.35-B70-T2000.dat}\tbl

		\addplot+[mark repeat=80] table[x=time,y={model}] {\tbl};
		\addlegendentry{Model}

		\addplot+[mark repeat=80] table[x=time,y={fullPower}] {\tbl};
		\addlegendentry{Max.\ Power}

		\addplot+[mark repeat=80] table[x=time,y={adaptiveCondExpect}] {\tbl};
		\addlegendentry{Adaptive}

		\addplot+[mark repeat=80] table[x=time,y={constant10dB}] {\tbl};
		\addlegendentry{$\txpower=\SI{10}{\dB}$}

		\addplot+[mark repeat=80] table[x=time,y={constant}] {\tbl};
		\addlegendentry{Const.\ Budget}

	\end{axis}
\end{tikzpicture}\\
	\vspace{-.5em}
	\caption{%
		Resilience outage probability~$\resiloutageprob$ over time for a communication system with the parameters described in \autoref{sec:numerical-example}.
		\texttt{Model} refers to the trained \gls{rl} model, \texttt{Adaptive} is the algorithm introduced in \autoref{sub:adaptive}, and the remaining curves correspond to constant powers at different levels.
		The shown values are averaged over \num{2000}~\gls{mc} runs.}
	\label{fig:results-resilience-outage-prob}
\end{figure}

Since the budget is related to the resilience outage probability~$\resiloutageprob$ through the alert outage probability~$\alertoutprob$, it is difficult to estimate the resilience performance directly from \autoref{fig:results-budget}.
Therefore, we show the resilience outage probabilities~$\resiloutageprob$ over time in \autoref{fig:results-resilience-outage-prob}.
First, we can see that the maximum power allocation is a lower bound on the resilience outage probability as expected.
For this, the outage probability stabilizes at around~$10^{-2}$, which is consistent with the results derived in \autoref{sub:const-power}.
Next, it can be seen that the \gls{ml}-based \texttt{Model} scheme achieves the second highest resilience, i.e., the second lowest resilience outage probability~$\resiloutageprob$ which stabilizes at around~\num{0.0173}.
This shows that the \gls{ml} approach can save significant power compared to the maximum power allocation as it only uses around \SI{1.5}{\percent} of the power in the steady state (\SI{11.9}{\dB} vs.\ \SI{30}{\dB}) while achieving only a slightly higher resilience outage probability (\num{0.0173} vs.\ \num{0.01}).
Another interesting observation from \autoref{fig:results-resilience-outage-prob} is that the resilience outage probability approaches~\num{1} for the \texttt{Const.~Budget} scheme, even though the budget stays above the minimum budget~$\minbudget$ on average.
In contrast, the resilience outage probability for the \texttt{Adaptive} algorithm stabilizes at around~\num{0.035} far below one, even though the average budget also approaches a finite value.
Furthermore, a similar long-term resilience performance is achieved by a constant power allocation with ${\txpower=\SI{10}{\dB}}$ with a significantly increasing average budget.
The reason for these unintuitive results is that \autoref{fig:results-budget} only shows the \emph{average} budget for each power allocation scheme.
However, since the resilience outage probability~$\resiloutageprob$ is a probabilistic measure, the distribution of the budget has a significant impact on it.
In the \texttt{Const.\ Budget} scheme, the constant transmit power is selected such that the average budget stays constant.
However, since the evolution of the budget over time is a random walk, the variance of the distribution of budgets increases over time.
This point is further illustrated in \autoref{fig:results-budget-quantiles-const-budget}, where different quantiles of the distribution of the budget over time are shown for the \texttt{Const.\ Budget} scheme.
As expected, the average budget stays constant over time.
In contrast, the spread of the distribution increases, e.g., the \SI{97}{\percent} quantile increases over time while the \SI{3}{\percent} quantile hits zero after around 80~time slots.
With this plot, we can also better understand the results shown in \autoref{fig:results-resilience-outage-prob}, in particular why the resilience outage probability increases to values close to one over time.
The resilience outage probability~$\resiloutageprob$ and the budget~$\budget$ are connected through the alert outage probability~$\alertoutprob$, and it generally holds that~$\resiloutageprob$ increases if the probability of having a small budget increases.
This effect can be seen in \autoref{fig:results-budget-quantiles-const-budget} as even higher quantiles, e.g., the \SI{60}{\percent} and \SI{70}{\percent} quantiles decrease over time or even hit zero.
In particular, the \SI{60}{\percent} quantile reaches zero after around \num{1300}~time slots, i.e., in \SI{60}{\percent} of the cases, the system completely runs out of key bits after around \num{1300}~time slots which directly translates to a large resilience outage probability.
In contrast to this, the distribution of the budget for the \texttt{Adaptive} power allocation scheme is narrower such that even lower quantiles stay close to the average and above the minimum required budget, resulting in a smaller resilience outage probability.
Additionally, having a narrower budget distribution can be beneficial in the case that the memory for storing key bits is limited.

\begin{figure}
	\centering
	\begin{tikzpicture}
	\begin{axis}[
		betterplot,
		xlabel={Time~$t$},
		ylabel={Budget~${\budget}$ [$\si{\bit}$]},
		xmin=0,
		xmax=2000,
		ymin=0,
		ymax=600,
		legend pos=north west,
		extra y ticks={40},
		extra y tick labels={$\minbudget$},
		]
		\pgfplotstableread[x=time]{data/2024-03-13-01-25-PPO-2/budget-p0.35-B70-T2000.dat}\tbl
		\pgfplotstableread[x=time]{data/budget_quant-991-p0.35-B70-T2000.dat}\tba
		\pgfplotstableread[x=time]{data/budget_quant-982-p0.35-B70-T2000.dat}\tbb
		\pgfplotstableread[x=time]{data/budget_quant-973-p0.35-B70-T2000.dat}\tbc
		\pgfplotstableread[x=time]{data/budget_quant-9010-p0.35-B70-T2000.dat}\tbd
		\pgfplotstableread[x=time]{data/budget_quant-7525-p0.35-B70-T2000.dat}\tbe
		\pgfplotstableread[x=time]{data/budget_quant-6040-p0.35-B70-T2000.dat}\tbf
		\pgfplotstableread[x=time]{data/budget_quant-7030-p0.35-B70-T2000.dat}\tbg

		\addplot[mark repeat=80,mark=pentagon*, thick, mark options=solid,fill=plot4,draw=black] table[x=time,y={constant}] {\tbl};
		\addlegendentry{Average Budget}

		\addplot[name path=constantLower3,plot2,forget plot] table[x=time, y={constant3}] {\tbc};
		\addplot[name path=constantUpper97,plot2,forget plot] table[x=time, y={constant97}] {\tbc};
		\addplot[plot2,opacity=.15] fill between[of=constantLower3 and constantUpper97];
		\addlegendentry{\SIrange{3}{97}{\percent} quantile range}

		\addplot[name path=constantLower10,plot1,forget plot] table[x=time, y={constant10}] {\tbd};
		\addplot[name path=constantUpper90,plot1,forget plot] table[x=time, y={constant90}] {\tbd};
		\addplot[plot1,opacity=.15] fill between[of=constantLower10 and constantUpper90];
		\addlegendentry{\SIrange{10}{90}{\percent} quantile range}

		\addplot[name path=constantLower30,plot0,forget plot] table[x=time, y={constant30}] {\tbg};
		\addplot[name path=constantUpper70,plot0,forget plot] table[x=time, y={constant70}] {\tbg};
		\addplot[plot0,opacity=.15] fill between[of=constantLower30 and constantUpper70];
		\addlegendentry{\SIrange{30}{70}{\percent} quantile range}

		\addplot[name path=constantLower40,plot3,forget plot] table[x=time, y={constant40}] {\tbf};
		\addplot[name path=constantUpper60,plot3,forget plot] table[x=time, y={constant60}] {\tbf};
		\addplot[plot3,opacity=.15] fill between[of=constantLower40 and constantUpper60];
		\addlegendentry{\SIrange{40}{60}{\percent} quantile range}

	\end{axis}
\end{tikzpicture}
	\vspace{-.5em}
	\caption{Multiple quantiles and average of the budget~$\budget$ over time for the \texttt{Const.~Budget} power allocation scheme.}
	\label{fig:results-budget-quantiles-const-budget}
\end{figure}

\section{Conclusion}\label{sec:conclusion}
We have considered a wireless communication system, in which \gls{skg} is performed to generate key bits.
These key bits are used as one-time pads to protect messages from a passive eavesdropper.
It is therefore important to ensure that the system always has a sufficient amount of key bits available.
In this work, we have proposed resilience metrics for such communication systems with an \gls{sk} budget.
These metrics are an important step for quantifying the physical layer resilience of modern communication systems.
While the resilience outage probability presented in this work is specifically tailored to the survivability of a system with \gls{sk} budget, it can be used as a basis for exploring resilience on the physical layer in the future.

Furthermore, we have proposed multiple power allocation schemes and analyzed them.
For a constant transmit power, we have provided lower and upper bounds on the resilience outage probability, analyzed its behavior over time, and derived the long-term convergence.
Additionally, we have provided insights into the influence of the transmit power and how it can be minimized while ensuring a given resilience requirement.
With adaptive and machine learning-based solutions, we have shown that dynamic power allocation can reduce the consumed power while simultaneously achieving a similar or better performance than constant power allocation.

In this work, we have only focused on the normal operation state and the preparation for possible alert states.
In future work, it is therefore of interest to consider the other resilience phases and take the frequency of alert states into account.
Another interesting extension of this work will be considering correlation and general dependency structures~\cite{Besser2021tcom} between the legitimate and eavesdropper's channels.

\appendices
\section{Proof of \autoref{thm:bounds-resilience-outage-prob}}\label{app:proof-thm-bounds-resilience-outage-prob}
The probability of \emph{not} violating the resilience target is given as
\begin{align*}
	\Pr\left(\alertoutprob(t) \leq \alertoutprobmax\right) 
	&= \Pr\left(\budget(t) \geq \minbudget,\; \min_{0 \leq i \leq t}\budget(i) > 0\right)\\
	&= \Pr\left(\sum_{i=1}^{t} \netusage(i) \leq \initbudget-\minbudget,\; \min_{0 \leq i \leq t}\budget(i) > 0\right),
\end{align*}
which is the joint probability that the budget available at time~$t$ is larger than the required minimum~$\minbudget$ and that the system survived up to time~$t$.

This joint probability can be bounded by the Fr\'{e}chet bounds~\cite{Ruschendorf1981} as
\begin{align}
	\begin{split}
		\Pr&\left(\alertoutprob(t) \leq \alertoutprobmax\right)\\
		& \geq \positive{\Pr\big(\netsum(t) \leq \initbudget-\minbudget\big) + \underbrace{\Pr\left(\min_{0 \leq i \leq t}\budget(i)>0\right)}_{\probsurv_{\initbudget}(t)}-1}
	\end{split}\\
	&= \positive{\cdf_{\netsum(t)}\big(\initbudget-\minbudget\big) - \probruin_{\initbudget}(t)}%
\end{align}
and
\begin{align}
	\begin{split}
		\Pr&\left(\alertoutprob(t) \leq \alertoutprobmax\right)\\
		& \leq \min\left\{\Pr\big(\netsum(t) \leq \initbudget-\minbudget\big),\; \Pr\left(\min_{0 \leq i \leq t}\budget(i)>0\right) \right\}
	\end{split}\\
	&= \min\left\{\cdf_{\netsum(t)}\big(\initbudget-\minbudget\big),\; \probsurv_{\initbudget}(t)\right\}\,,
\end{align}
where we use the shorthand $\netsum(t) = \sum_{i=1}^{t} \netusage(i)$, and the fact that $\probruin_{\initbudget}(t) = 1-\probsurv_{\initbudget}(t)$.

With the relation $\Pr\left(\alertoutprob(t) \leq \alertoutprobmax\right) = 1-\Pr\left(\alertoutprob(t) > \alertoutprobmax\right)$, we obtain the statement of the theorem.

\section{Proof of \autoref{thm:long-term-resilience-outage}}\label{app:proof-cor-long-term-resilience-outage}
For $\probtx>\probtxcrit$, it follows from \autoref{lem:expected-net-claim-random} that $\expect{\netusage}>0$, i.e., more key bits are on average used in each time slot than generated.
From \cite[{Cor.~2}]{Besser2024SKGbudget}, it follows that the system will run out of \gls{sk} bits almost surely, i.e.,
\begin{equation*}
	\lim\limits_{t\to\infty}\probruin_{\initbudget}(t) = 1\,.
\end{equation*}
If we combine this with the bounds on the resilience outage probability~$\resiloutageprob$ from \autoref{thm:bounds-resilience-outage-prob}, it can be seen that both the lower bound and the upper bound converge to \num{1}, i.e., the actual resilience outage probability also converges to \num{1},
\begin{equation*}
	\lim\limits_{t\to\infty} \resiloutageprob(t) = 1\,.
\end{equation*}

Similarly, we have $\expect{\netusage}<0$ for $\probtx<\probtxcrit$.
From \cite[{Chap.~XII.2, Thm.~1}]{Feller1991}, it follows that the random walk $\netsum=\sum_{i=1}^{t}\netusage(i)$ drifts to $-\infty$.
Thus, the probability that $\netsum(t)$ is less than the finite value~$\initbudget-\minbudget$ approaches \num{1}, i.e.,
\begin{equation}
	\lim\limits_{t\to\infty} \Pr\left(\netsum(t)\leq\initbudget-\minbudget\right) = 1\,.
\end{equation}
Combining this with the bounds from \autoref{thm:bounds-resilience-outage-prob} yields
\begin{equation}
	\lim\limits_{t\to\infty} \probruin_{\initbudget}(t) \leq \lim\limits_{t\to\infty} \resiloutageprob(t) \leq \lim\limits_{t\to\infty} \probruin_{\initbudget}(t)\,,
\end{equation}
which corresponds to \eqref{eq:resil-outage-prob-long-term-small-p}.

\printbibliography[heading=bibintoc]

\end{document}